\documentclass[12pt]{amsart}
\usepackage{amsmath,amsfonts,latexsym,graphicx,amssymb,url}
\usepackage{amsmath,txfonts,pifont,bbding,pxfonts,manfnt}
\usepackage{psfrag}
\usepackage{color}
\usepackage{mathrsfs}
\usepackage{wrapfig, subfigure,graphicx}
\usepackage{hyperref}
\usepackage[active]{srcltx}
\usepackage{breqn}
\usepackage{pdfsync}
\usepackage{pdflscape}
\newcommand{\R}{{\mathbb R}} 

\newcommand{\E}{{\mathbf E}}
\newcommand{\B}{{\mathbf H}}
\def \S {{\mathbf S}}
\def \r {{\mathbf r}}

\theoremstyle{plain}
\newtheorem{theorem}{Theorem}[section]

\newtheorem{proposition}[theorem]{Proposition}

\newtheorem{remark}[theorem]{Remark}

\begin{document}

\setcounter{equation}{0}










\title[Uniform Refraction in NIMs]
{Uniform Refraction in Negative Refractive Index Materials}
\author[C. E. Guti\'errez and E. Stachura]
{Cristian E. Guti\'errez and Eric Stachura}
\thanks{\today\\The first author is partially supported
by NSF grant DMS--1201401.}
\address{Department of Mathematics\\Temple University\\Philadelphia, PA 19122}
\email{gutierre@temple.edu}
\address{Department of Mathematics\\Temple University\\Philadelphia, PA 19122}
\email{eric.stachura@temple.edu}

\maketitle
\begin{abstract}
We study the problem of constructing an optical surface separating two homogeneous, isotropic media, one of which has a negative refractive index. In doing so, we develop a vector form of Snell's law, which is used to study surfaces possessing a certain uniform refraction property, both in the near and far field cases. In the near field problem, unlike the case when both materials have positive refractive index, we show that the resulting surfaces can be neither convex nor concave. 
\end{abstract}
\tableofcontents

\section{Introduction}
Given a point $O$ inside medium $I$ and a point $P$ inside medium $II$, the question we consider is whether we can find an interface surface $\Gamma$ separating medium $I$ and medium $II$ that refracts all rays emanating from the point $O$ into the point $P$. Suppose that medium $I$ has index of refraction $n_1$ and medium $II$ has index of refraction $n_2$. In the case that both $n_1, n_2 > 0$, this question has been studied extensively (see in particular \cite{CGQH1} and the references therein). This is the so called near field problem. The far field problem is as follows: given a direction $m$ fixed in the unit sphere in $\R^3$, if rays of light emanate from the origin inside medium I, what is the surface $\Gamma$, interface between media I and II, that refracts all these rays into rays parallel to $m$? In the case that both $n_1, n_2 > 0$ this problem has also been studied extensively \cite{CGQH2}. 

A natural question in this direction is what happens in the case when one of the media has a negative index of refraction? These media, as of now still unknown to exist naturally, constitute the so called Negative Refractive Index Materials (NIM) or left-handed materials. The theory behind these materials was developed by V. G. Veselago in the late 1960's \cite{Ves1968}, yet was put on hold for more than 30 years until the existence of such materials was shown by a group at the University of California at San Diego \cite{Shelby2001}.

This note is devoted to studying surfaces having the uniform refraction property (as discussed above) in the case when $\kappa := n_2 / n_1 < 0$, both in the far field and near field cases. We begin discussing in Section \ref{subsect:intronegrefractivematerials} the notion of negative refraction, and in Section \ref{subject:snelllawinvectorform} we formulate Snell's law in vector form for NIM.
Next, in Section \ref{sec:surfacesuniformrefraction} we find the surfaces having  the uniform refraction property in the setting of NIMs. In contrast with standard materials, i.e. those having positive refractive indices, the surfaces for NIMs in the near field case can neither be convex nor concave; this is analyzed in detail in Section \ref{subsect:convexityofsurfaces}.

Finally, in Section \ref{sec:maxwellandfresnelformulas} using the Snell law described in Section \ref{subject:snelllawinvectorform}, we study the Fresnel formulas of geometric optics, which gives the amount of radiation transmitted and reflected through a surface separating two homogeneous and isotropic media, one of them being a NIM medium.  
We expect that the study in this paper will be useful in the design of surfaces refracting energy in a prescribed manner for NIMs as it is done for standard materials in \cite{CGQH1}, \cite{CGQH2}, and \cite{CG2013}.

%
%

\section{Geometric Optics in Negative Refractive Index Materials}

\subsection{Introduction to Negative Refractive Index Materials}\label{subsect:intronegrefractivematerials}
The notion of negative refraction goes back to the work of V. G. Veselago in \cite{Ves1968}. In the Maxwell system of electromagnetism, the material parameters $\epsilon, \mu$ are characteristic quantities which determine the propagation of light in matter. If $\textbf{E} = A \cos( \textbf{r} \cdot \textbf{k} + \omega t)$, then in the case of isotropic material the dispersion relation \cite[Formula (7.9)]{jackson:electrodynamics} is given by 

$$|{\bf k}|^2 = \left(\dfrac{\omega}{c}\right)^2 n^2$$ with $n^2$ the square of the index of refraction of the material, which takes the form

\begin{align} \label{eq:squared}
n^2 = \epsilon \mu
\end{align}
Veselago showed that in the case when both $\epsilon , \mu < 0$, we are forced by the Maxwell system to take the negative square root for the refractive index, i.e. 

\begin{align} \label{negsqroot}
n = -\sqrt{\epsilon \mu}
\end{align}
In particular he observed that a slab of material with

\begin{align} \label{vesmaterial}
\epsilon = -1 \quad \text{ and } \quad \mu = -1
\end{align}
would have refractive index $n = -1$ and behave like a lens. 
A remarkable property of this material was shown in \cite{Pendry2000}: that the focusing is perfect provided the condition (\ref{vesmaterial}) is exactly met. 
The theory behind NIMs opened the door to study the so called \emph{metamaterials}, see e.g. \cite{RGbook} and \cite{Metabook}.

 


We end this section by making some remarks on some important differences in geometric optics in the case of NIMs. One must generalize the classical Fermat principle of optics, see e.g. \cite[Sec. 3.3.2]{BW1959}, to handle materials with negative refractive index. This was done not too long ago by Veselago in \cite{Ves2002}. In particular, in the setting of NIMs, the optical length of a light ray propagating from one NIM to another is negative because the wave vector is opposite to the direction of travel of the ray. However, in this case, it is not possible a priori to determine that the path of light is a maximum or minimum of the optical length, as this depends heavily on the geometry of the problem. Finally, the Fresnel formulas of geometric optics must also be appropriately adjusted for negative refraction; we analyze this in detail at the end.

\subsection{Snell's law in vector form for $\kappa<0$}\label{subject:snelllawinvectorform}
In order to explain the Snell law of refraction for media with $\kappa < 0$, we first review this for media with positive refractive indices.

Suppose $\Gamma$ is a surface in $\R^3$ that separates two media
$I$ and $II$ that are homogeneous and isotropic, with refractive indices $n_1$ and $n_2$ respectively.
If a
ray of light\footnote{Since the refraction angle depends on the frequency of the radiation, we assume that light rays are monochromatic.} having direction $x\in S^{2}$, the unit sphere in $\R^{3}$, and traveling
through medium $I$ strikes $\Gamma$ at the point $P$, then this ray
is refracted in the direction $m\in S^{2}$ through medium $II$
according to the Snell law in vector form:
\begin{equation}\label{snellwithcrossproduct}
n_{1}(x\times \nu)=n_{2}(m\times \nu),
\end{equation} 
where $\nu$ is the unit normal to the surface to $\Gamma$ at $P$ pointing towards medium $II$; see \cite[Subsection 4.1]{luneburgoptics}. It is assumed here that $x\cdot \nu\geq 0$.

This has several consequences:
\begin{enumerate}
\item[(a)] the vectors $x,m,\nu$ are all on the same plane (called the plane of incidence);
\item[(b)] the well known Snell's law in scalar form holds: 
$$n_1\sin \theta_1= n_2\sin
\theta_2,$$ 
where $\theta_1$ is the angle between $x$ and $\nu$
(the angle of incidence), and
$\theta_2$ is the angle between $m$ and $\nu$ (the angle of refraction).
\end{enumerate}
Equation \eqref{snellwithcrossproduct} is equivalent to $(n_{1}x-n_{2}m)\times \nu=0$, which means that the
vector $n_{1}x-n_{2}m$ is parallel to the normal vector $\nu$.
If we set $\kappa=n_2/n_1$, then
\begin{equation}\label{eq:snellvectorform}
x-\kappa \,m =\lambda \nu,
\end{equation}
for some $\lambda\in \R$. Notice that \eqref{eq:snellvectorform} univocally determines $\lambda$. Taking dot products with $x$ and $m$ in \eqref{eq:snellvectorform} we get
$\lambda=\cos \theta_1-\kappa \cos \theta_2$,
$\cos \theta_1=x\cdot \nu>0$, and
$\cos \theta_2=m\cdot \nu=\sqrt{1-\kappa^{-2}[1-(x\cdot \nu)^2]}$. In fact, there holds
\begin{equation}\label{eq:formulaforlambda}
\lambda=x\cdot \nu -\kappa \,\sqrt{1-\kappa^{-2}\left(1-(x\cdot \nu)^{2}\right)}.
\end{equation}

It turns out that refraction behaves differently for $0<\kappa<1$ and for $\kappa>1$.
Indeed,
\begin{enumerate}
\item\label{item:conditionkappa<1} if $0<\kappa<1$ then for refraction to occur we need $x\cdot \nu \geq \sqrt{1-\kappa^2}$, and in \eqref{eq:formulaforlambda} we have $\lambda>0$, and the refracted vector $m$ is so that $x\cdot m\geq \kappa$;
\item\label{item:conditionkappa>1} if $\kappa>1$ then refraction always occurs and we have $x\cdot m\geq 1/\kappa$, and in \eqref{eq:formulaforlambda} we have $\lambda<0$;
\end{enumerate}
see \cite[Subsection 2.1]{CGQH2}. 

We now consider the case when either medium $I$ or medium $II$ has negative refractive index, so that $\kappa=\dfrac{n_2}{n_1}<0$. Let us take the formulation of the Snell law as in  \eqref{eq:snellvectorform}.
This is again equivalent to
\begin{equation}\label{eq:snelllawnegative}
x-\kappa \,m=\lambda\,\nu,
\end{equation}
with $x,m$ unit vectors, $\nu$ the unit normal to the interface pointing \emph{towards medium $II$}. We will show that
\begin{equation}\label{eq:formulaforlambdakappanedagtive}
\lambda = x \cdot \nu + \sqrt{(x\cdot \nu)^2 - (1 - \kappa^2)}.
\end{equation}
In fact, taking dot products in \eqref{eq:snelllawnegative}, first with $x$ and then with $m$, yields
$$1 - \kappa\, x \cdot m = \lambda \,x \cdot \nu,\quad \text{and}\quad x \cdot m - \kappa = \lambda \,m \cdot \nu.$$ 
This implies that
$\dfrac{1}{\kappa} - \dfrac{1}{\kappa} \lambda x \cdot \nu = x \cdot m = \kappa + \lambda m \cdot \nu.$
We seek to get rid of the term $m \cdot \nu$. To do this, from \eqref{eq:snelllawnegative} we have that
$m = \dfrac{x - \lambda \nu}{\kappa}$, so 
$m \cdot \nu = \dfrac{1}{\kappa} x\cdot \nu - \dfrac{1}{\kappa} \lambda.$
By substitution, we therefore obtain the following quadratic equation in $\lambda$:
$$\lambda^2 - 2 \lambda x \cdot \nu + (1 - \kappa^2) = 0. $$ 
Solving this equation yields
$$\lambda = x \cdot \nu \pm \sqrt{ (x\cdot \nu)^2 -  (1 - \kappa^2)}.$$ 
If $-\infty<\kappa\leq -1$, then $(x\cdot \nu)^2 -  (1 - \kappa^2)\geq 0$ and the square root is real. On the other hand, if $-1<\kappa <0$, then the square root is real only if $x\cdot \nu\geq \sqrt{1-\kappa^2}$ which means that the incident direction $x$ must satisfy this condition.
It remains to check which sign $(\pm)$ to take for $\lambda$. Recalling that
$$x \cdot \nu - \kappa m \cdot \nu = \lambda = x \cdot \nu \pm \sqrt{(x \cdot \nu)^2 - (1 - \kappa^2)},$$
we see, since $\kappa < 0$ and $m\cdot \nu\geq 0$, that we must take the positive square root. Hence we conclude
\eqref{eq:formulaforlambdakappanedagtive}.
Notice that in contrast with the case when $\kappa>0$, the value of $\lambda$ given in \eqref{eq:formulaforlambdakappanedagtive} is always positive when $\kappa<0$, and it can be also written in the following from 
\[
\lambda=x\cdot \nu +|\kappa| \,\sqrt{1-\kappa^{-2}\left(1-(x\cdot \nu)^{2}\right)},
\]
which is similar to \eqref{eq:formulaforlambda}. In the "reflection" case $\kappa = -1$, since $x\cdot \nu>0$, 
this formula yields 
%
%
%
%
%
$$\lambda = 2 \,x \cdot \nu,$$
and hence the ``reflected" vector $m$ is given by
$$m = 2\, (x \cdot \nu)\, \nu - x.$$
That is, after striking the interface $\Gamma$ the wave with direction $m$ travels in the material with refractive index $n_2$. 

Note that the vector formulation \eqref{eq:snelllawnegative} is compatible with the Snell law for negative refractive index (\cite{Ves2003}).
Indeed, taking the cross product in \eqref{eq:snelllawnegative} with the normal $\nu$ yields
$x\times \nu=\kappa\,m\times \nu$, and then taking absolute values yields 
\[
\sin \theta_1=-\kappa \,\sin \theta_2,
\]
where $\theta_1$ is the angle between $x$ and $\nu$, and $\theta_2$ is the angle between $m$ and $\nu$.

\begin{figure}[h]
\begin{center}
\includegraphics[scale=0.5]{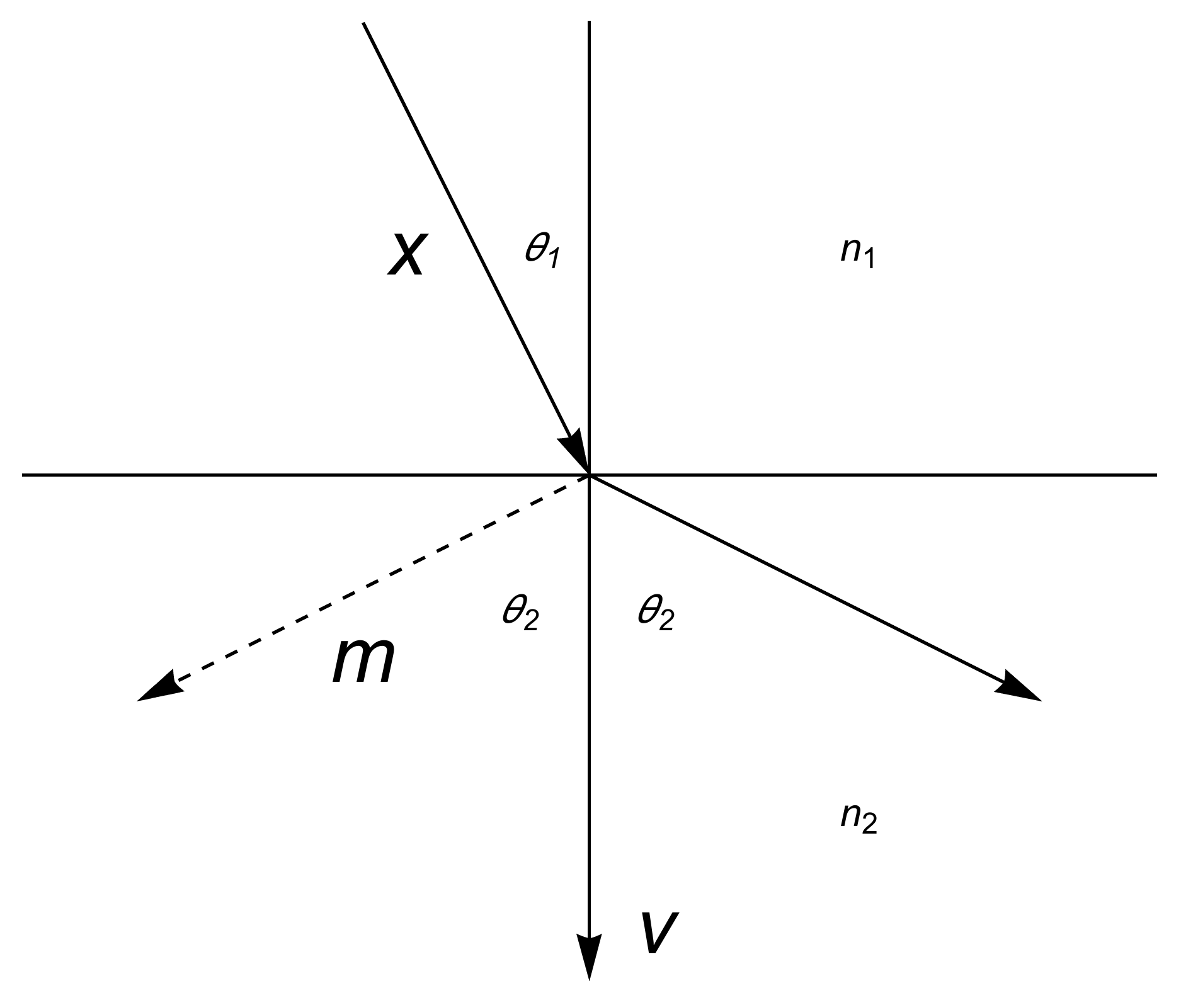}
\caption{Snell's law in the case of normal refraction and negative refraction (dotted arrow).}
\label{Snell}
\end{center}
\end{figure}

 
\subsection{Physical condition for refraction when $\kappa<0$}

We analyze here conditions under which an electromagnetic wave is transmitted from medium $I$ to $II$ and there is no total internal reflection.

\subsubsection{Case when $-1\leq \kappa <0$.}\label{physicalconditionforrefractionkappa<0}\label{subsubsec:physicalconditionforrefraction}

The maximum angle of refraction is $\theta_2=\pi/2$ which is attained when
$\sin \theta_1=-\kappa$, that is, the critical angle is $\theta_c=\arcsin (-\kappa)$. If $\theta_1>\theta_c$, there is no refraction.
So $\theta_2=\arcsin \left( -\dfrac{1}{\kappa}\sin \theta_1\right)$, for $0\leq \theta_1\leq \theta_c$.
The dot product $x\cdot m=\cos (\theta_1+\theta_2)$.
Let $h(\theta_1)=\theta_1+\theta_2=\theta_1+\arcsin \left( -\dfrac{1}{\kappa}\sin \theta_1\right)$.
We have
\[
h'(\theta_1)=1-\dfrac{1}{\kappa}\cos \theta_1\,\dfrac{1}{\sqrt{1-\sin^2\theta_1/\kappa^2}}\geq 0
\]
for $0\leq \theta_1\leq \theta_c$, and so $h$ is increasing on $[0,\theta_c]$ and therefore 
\[
\theta_1+\theta_2\leq \theta_c+\pi/2,\quad
\text{for $0\leq \theta_1\leq \theta_c$}.
\]
Then the physical constraint for refraction is
\begin{equation}\label{eq:constraint_kappa_lessthan0}
x\cdot m=\cos (\theta_1+\theta_2)\geq \cos (\theta_c+\pi/2) = -\sin \theta_c =\kappa.
\end{equation}

\subsubsection{Case when $\kappa <-1$.}
In this case $0\leq \theta_1\leq \pi/2$, and since $\theta_2=\arcsin\left( -\dfrac{1}{\kappa}\sin \theta_1\right)$,
the maximum angle of refraction is when $\theta_1=\pi/2$, that is, the maximum angle of refraction is 
$\theta_r=\arcsin (-1/\kappa)$.
Now $h(\theta_1)$ is increasing on $[0,\pi/2]$, so $\theta_1+\theta_2\leq \pi/2 +\arcsin (-1/\kappa)$,
and 
\begin{equation}\label{eq:constraint_kappa_lessthan-1}
x\cdot m=\cos (\theta_1+\theta_2)\geq \cos \left( \pi/2 +\arcsin (-1/\kappa)\right)=\dfrac{1}{\kappa}.
\end{equation}

\section{Surfaces with Uniform Refraction Property}\label{sec:surfacesuniformrefraction}

\subsection{The Far Field Problem}
Let $m\in S^2$ be fixed. We ask the following question: 
if rays of light emanate from the origin inside medium $I$,
what is the surface $\Gamma$, interface between media $I$ and $II$,
that refracts all these rays into rays parallel to $m$?
Here $I$ has refractive index $n_1$, $II$ has refractive index $n_2$ such that $\kappa=n_2/n_1<0$.

This question can be answered using the vector form of Snell's law.
Suppose $\Gamma$ is parameterized by the polar representation
$\rho(x)x$ where $\rho>0$ and $x\in S^2$.
Consider a curve on $\Gamma$ given by $r(t)=\rho(x(t)) x(t)$
for $x(t)\in S^2$.
According to \eqref{eq:snelllawnegative}, the tangent vector $r'(t)$ to
$\Gamma$ satisfies $r'(t)\cdot (x(t)-\kappa \,m)=0$.
That is,
$\Big([\rho(x(t))]'x(t)+\rho(x(t)) x'(t)\Big)\cdot (x(t)-\kappa \,m)=0$,
which yields
$\left( \rho(x(t)) (1-\kappa m\cdot x(t))\right)'=0$. Therefore
\begin{equation}\label{eq:refractingellipsoid}
\rho(x)=\dfrac{b}{1-\kappa \,m\cdot x}
\end{equation}
for $x\in S^2$ and for some $b\in \R$.
To see what is the surface described by  \eqref{eq:refractingellipsoid},
we assume first that $-1<\kappa<0$. Since $m\cdot x\geq -1$, we have $1-\kappa\,m\cdot x\geq 1+\kappa>0$, 
and so $b>0$.

Suppose for simplicity that $m=e_3$, the third-coordinate vector.
If $y=(y',y_3)\in \R^3$ is a point on $\Gamma$, then
$y=\rho(x)x$ with $x=y/|y|$.
From \eqref{eq:refractingellipsoid}, $|y|-\kappa\,y_3=b$, that is,
$|y'|^2+y_3^2=(\kappa\,y_3+b)^2$ which yields
$|y'|^2 + (1-\kappa^2)y_3^2 -2\kappa b y_3=b^2$.
This equation can be written in the form
\begin{equation}\label{eq:ellipsoid}
\dfrac{|y'|^2}{\left(\dfrac{b}{\sqrt{1-\kappa^2}} \right)^2} +
\dfrac{\left( y_3 -\dfrac{\kappa b}{1-\kappa^2}\right)^2}
{\left(\dfrac{b}{1-\kappa^2} \right)^2}=1
\end{equation}
which is an ellipsoid of revolution about the $y_3$-axis with upper focus $(0,0)$, and lower focus
$\left(0,2\kappa b/(1-\kappa^2)\right)$.
Since $|y|=\kappa y_3+b$ and the physical constraint for refraction
\eqref{eq:constraint_kappa_lessthan0},
$\dfrac{y}{|y|}\cdot e_3\geq \kappa$ is equivalent to
$y_3\geq \dfrac{\kappa b}{1-\kappa^2}$. That is,
for refraction to occur $y$ must be in the upper half of the ellipsoid \eqref{eq:ellipsoid}; we denote this semi-ellipsoid by
$E(e_3,b)$.
To verify that $E(e_{3},b)$ has the uniform refracting property,
that is, it refracts any ray emanating from the origin in the direction $e_{3}$,
we check that \eqref{eq:snelllawnegative} holds at each point.
Indeed, if $y\in E(e_3,b)$, then $\left(\dfrac{y}{|y|}-\kappa e_3\right)\cdot \dfrac{y}{|y|}=
1-\kappa e_3\cdot \dfrac{y}{|y|}\geq 1+\kappa>0$,
and $\left(\dfrac{y}{|y|}-\kappa e_3\right)\cdot e_3\geq 0$,
and so $\dfrac{y}{|y|}-\kappa e_3$ is an outward normal to $E(e_3,b)$ at $y$.

%

\begin{figure}[ht]
\centering
\begin{minipage}[b]{0.45\linewidth}
\includegraphics[scale=0.2,angle=90]{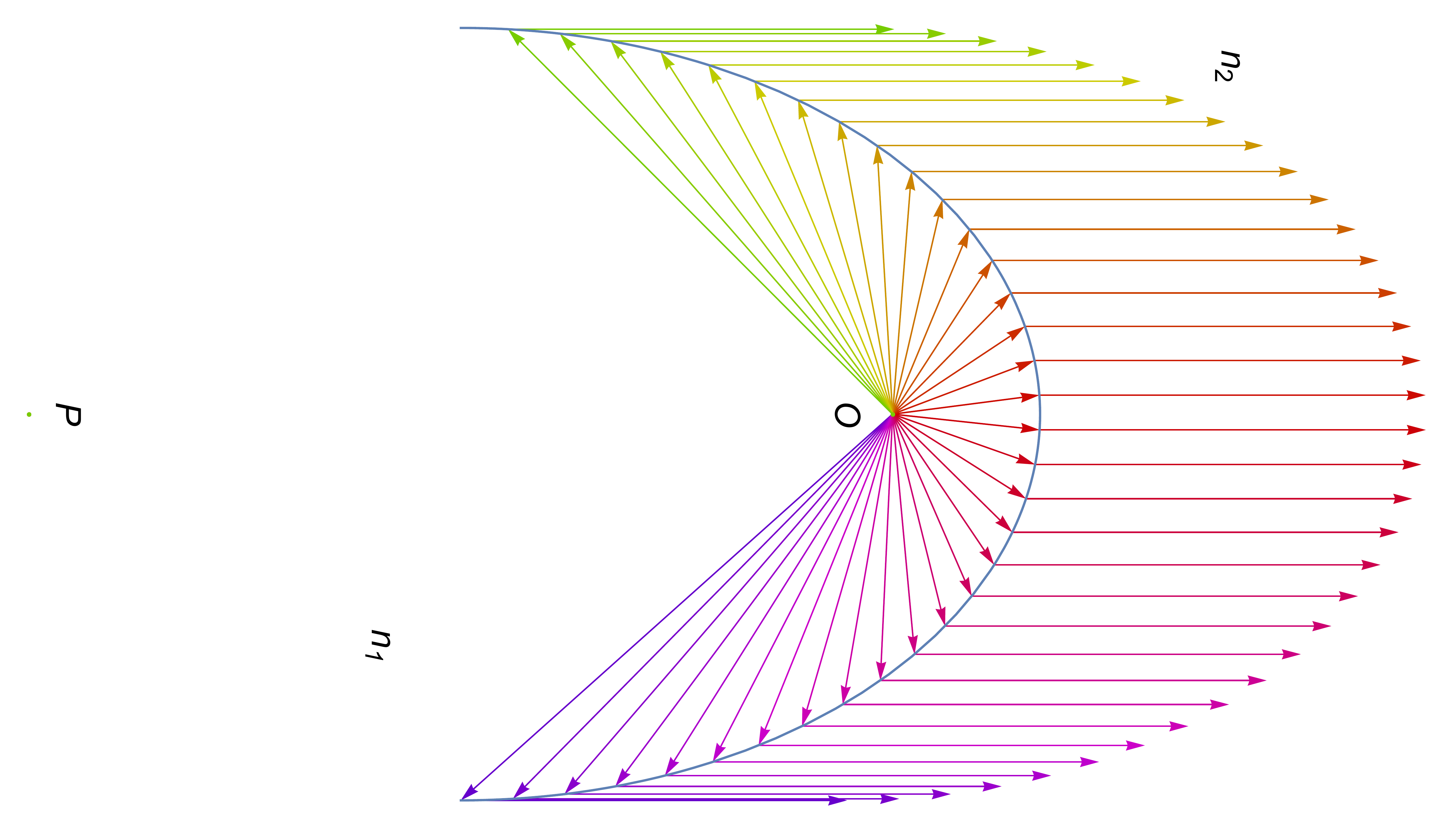}
\caption{Ellipsoid refracting when $-1<\kappa<0$}
\label{fig:negativeellipsoid}
\end{minipage}
\quad
\begin{minipage}[b]{0.45\linewidth}
\includegraphics[scale=0.2,angle=90]{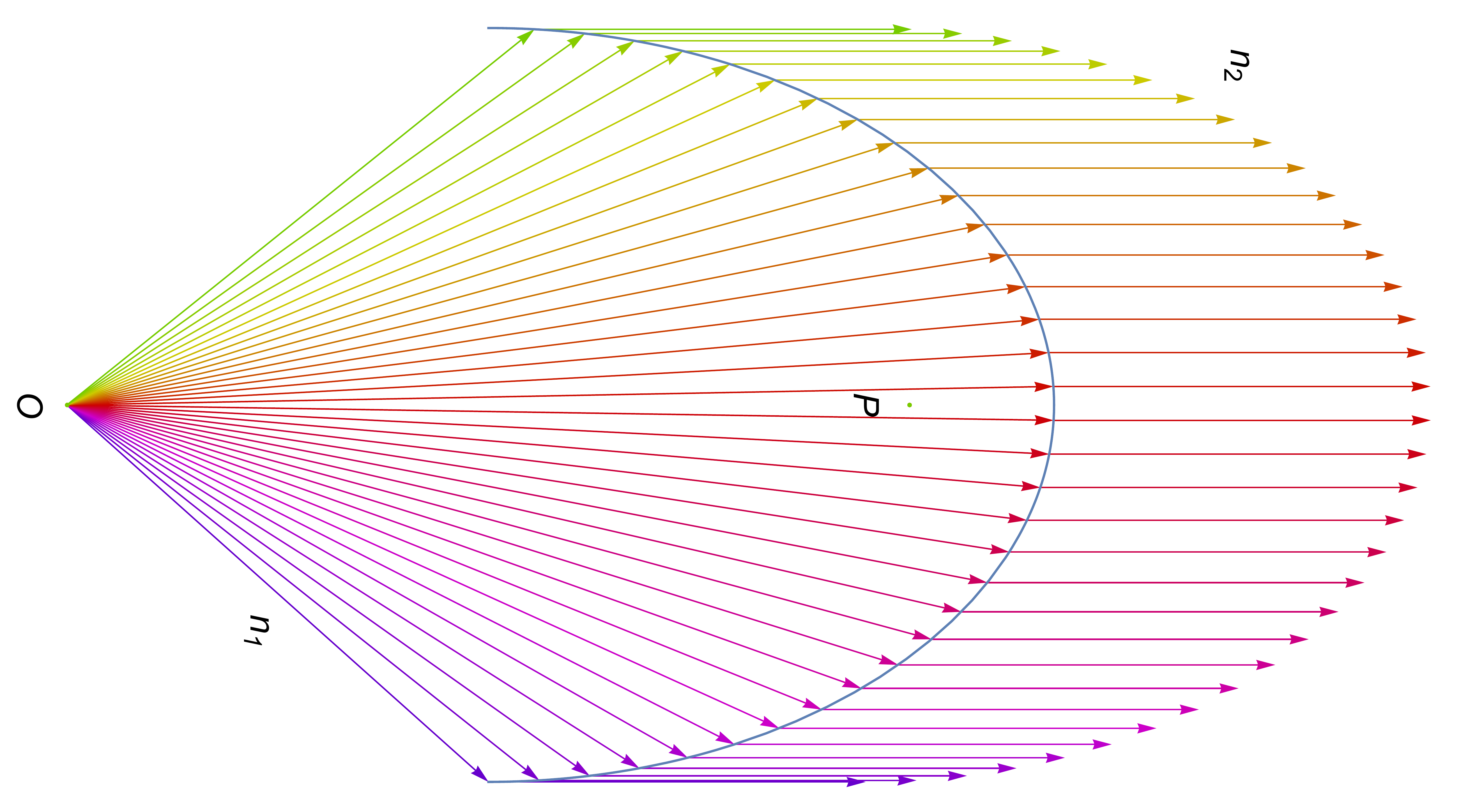}
\caption{Ellipsoid refracting when $0<\kappa<1$}
\label{fig:positiveellipsoid}
\end{minipage}
\end{figure}

Rotating the coordinates, it is easy to see that the surface given by
\eqref{eq:refractingellipsoid} with $-1<\kappa<0$ is an
ellipsoid of revolution about the axis of direction $m$ with upper focus
$(0,0)$ and lower focus $\dfrac{2\kappa b}{1-\kappa^2}m$.
Moreover, the semi-ellipsoid $E(m,b)$ given by
\begin{equation}\label{def:supportingellipsoid}
E(m,b)=\left\{\rho(x)x: \rho(x)= \dfrac{b}{1-\kappa \,m\cdot x},\
x\in S^2,\  x\cdot m\geq \kappa\right\},
\end{equation}
has the uniform refracting property: any ray emanating from
the origin $O$ is refracted in the direction $m$, where $-1<\kappa <0$.
See Figures \ref{fig:negativeellipsoid} and \ref{fig:positiveellipsoid} comparing refraction when $\kappa<0$ versus $\kappa>0$, where $P=\left(0,2\kappa b/(1-\kappa^2)\right)$ and $m=e_3$.

Let us now assume that $\kappa=-1$. 
From \eqref{eq:refractingellipsoid} we have $|y|+y_3=b$, and so $|y'|^2+y_3^2=b^2-2by_3+y_3^2$, and therefore
\[
y_3=\dfrac{1}{2b}\left( b^2-|y'|^2\right).
\]
This is a paraboloid with axis $e_3$ and focus at $O$.
Since $\kappa=-1$, from the physical constraint for refraction \eqref{eq:constraint_kappa_lessthan0} we get any ray with direction
$x\in S^2$ emanating from the origin is refracted by the paraboloid.

Now we turn to the case $\kappa<-1$.  Due to the physical constraint of
refraction \eqref{eq:constraint_kappa_lessthan-1}, we must have $b>0$ in \eqref{eq:refractingellipsoid}; we determine what kind of surface this is.
Define for $b>0$
\begin{equation}\label{eq:refractinghyperboloid}
H(m,b)=\left\{\rho(x)x: \rho(x)= \dfrac{b}{1-\kappa \,m\cdot x},\
x\in S^2, \ x\cdot m\geq 1/\kappa\right\}.
\end{equation}
We claim that $H(m,b)$ is the sheet with opening in direction $m$
of a hyperboloid of revolution of two sheets about the axis of direction $m$.
To prove the claim, set for simplicity $m=e_3$.
If $y=(y',y_3)\in H(e_3,b)$, then
$y=\rho(x)x$ with $x=y/|y|$.
From \eqref{eq:refractinghyperboloid}, $-\kappa\,y_3+|y|=b$,
and therefore
$|y'|^2+y_3^2=(\kappa\,y_3+b)^2$ which yields
$|y'|^2 - (\kappa^2-1)\left[ \left(y_3 +
\dfrac{\kappa b}{\kappa^2 -1}\right)^2 -
\left( \dfrac{\kappa b}{\kappa^2-1}\right)^2 \right]=b^2$.
Thus, any point $y$ on $H(e_3,b)$ satisfies the equation
\begin{equation}\label{eq:cartisianhyperboloid}
\dfrac{\left( y_3 +\dfrac{\kappa b}{\kappa^2-1}\right)^2}
{\left(\dfrac{b}{\kappa^2-1} \right)^2}
-
\dfrac{|y'|^2}{\left(\dfrac{b}{\sqrt{\kappa^2-1}} \right)^2}
=1
\end{equation}
which represents a hyperboloid of revolution of two sheets
about the $y_3$ axis with foci
$(0,0)$ and $(0,-2\kappa b/(\kappa^2-1))$.
The sheets of this hyperboloid of revolution are given by
\[
y_3= -\dfrac{\kappa b}{\kappa^2 -1} \pm
\dfrac{b}{\kappa^2-1}\sqrt{1 + \dfrac{|y'|^2}
{\left(b/\sqrt{\kappa^2-1}\right)^2}}.
\]
We decide which one satisfies the refracting property.
The sheet with the minus sign in front of the square root satisfies 
$\kappa y_3+b\geq -\dfrac{b}{\kappa-1}>0$, and hence has polar equation
$\rho(x)=\dfrac{b}{1-\kappa\, e_3\cdot x }$, and therefore this is the sheet to consider.
For a general $m$, by a rotation, we obtain that $H(m, b)$ is the sheet
with opening in direction opposite to $m$
of a hyperboloid of revolution of two sheets about the axis
of direction $m$ with foci
$(0,0)$ and $\dfrac{-2\kappa b}{\kappa^2-1}m$.

Notice that the focus $(0,0)$ is inside the region enclosed by $H(m,b)$ and the focus
$\dfrac{-2\kappa b}{\kappa^2-1}m$ is outside that region.
The vector $\dfrac{y}{|y|}-\kappa m$ is an outer normal to $H(m,b)$ at $y$,
since by \eqref{eq:refractinghyperboloid}
\begin{align*}
\left(\dfrac{y}{|y|}-\kappa m\right)\cdot
\left(\frac{-2\kappa b}{\kappa^2-1}m-y\right)&
\geq \frac{-2 b}{\kappa^2-1}+\frac{2\kappa^2 b}{\kappa^2-1}
+\kappa m\cdot y-|y|  \\
&=\frac{-2 b}{\kappa^2-1}+\frac{2\kappa^2 b}{\kappa^2-1}-b=b>0.
\end{align*}
Clearly, $\left(\dfrac{y}{|y|}-\kappa m\right)\cdot m\geq \dfrac{1}{\kappa}-\kappa>0$ and
$\left(\dfrac{y}{|y|}-\kappa m\right)\cdot \dfrac{y}{|y|}>0$.
Therefore, $H(m,b)$ satisfies the uniform refraction property. We summarize the uniform refraction property below. 

\begin{theorem} \label{thm:uniformrefractionfarfield}
Let $n_1, n_2$ be two indices of refraction for media $I$ and $II$, respectively, and set $\kappa = n_2 /n_1$. Assume the origin is in medium $I$, and $E(m, b)$, $H(m,b)$ defined by \eqref{eq:refractingellipsoid} and \eqref{eq:refractinghyperboloid} respectively. Then:

\begin{enumerate}
\item If $-1 < \kappa < 0$ and $E(m,b)$ is the interface between medium $I$ and medium $II$, then $E(m,b)$ refracts all rays from the origin $O$ into rays in medium $II$ with direction $m$.
\item If $\kappa < -1$ and $H(m,b)$ is the interface between medium $I$ and medium $II$, then $H(m,b)$ refracts all rays from the origin $O$ into rays in medium $II$ with direction $m$. 
\end{enumerate}

\end{theorem}

\subsection{The Near Field Problem}
Given a point $O$ inside medium $I$ and a point $P$ inside medium $II$, we construct a surface $\Gamma$ separating medium $I$ and medium $II$ such that $\Gamma$ refracts all rays emanating from $O$ into the point $P$. 
Suppose $O$ is the origin, and let $X(t)$ be a curve on $\Gamma$. Recall Snell's law says that
$$x - \kappa m = \lambda \nu$$ where $\nu$ is the normal to the surface $\Gamma$, and we are considering the case when $\kappa < 0$. Then via Snell's law, we see that the tangent vector $X'(t)$ must satisfy
$$X'(t) \cdot \left( \frac{X(t)}{|X(t)|} - \kappa \frac{P-X(t)}{|P - X(t)|} \right) = 0.$$
That is to say,
$$|X(t)|' + \kappa|P - X(t)|' = 0.$$
Hence $\Gamma$ is the surface
\begin{equation}\label{eq:ovalnegativek}
|X| + \kappa |X - P| = b.
\end{equation}

Compare this surface to the Cartesian ovals which occur in the near field case when $\kappa > 0$. In this case, since the function $F(X) = |X| + \kappa |X - P|$ is convex and the ovals are the level sets of $F$, then the ovals are convex. However, in the case when $\kappa < 0$, the function $F(X)$ is no longer convex and its level sets are convex sets only for a range of values in the parameter $b$. That is, the surface \eqref{eq:ovalnegativek} is in general not a convex surface, see Theorem \ref{klessthan0} and Figure \ref{fig:ovalskappa=-.5}.

{\bf Case $-1 < \kappa < 0$.} We will show that the surface \eqref{eq:ovalnegativek} is 
nonempty if and only if   
$b$ is bounded from below.
In fact, this follows by the triangle inequality; since $1+\kappa>0$, we see that 
$|X|+\kappa \,|X-P|\geq \kappa\,|P|$ for all $X$.
Therefore if the surface \eqref{eq:ovalnegativek} is nonempty then
\begin{align} \label{bgeqkappap}
b \geq \kappa |P|.
\end{align}
Viceversa, if $b$ satisfies \eqref{bgeqkappap}, then the surface \eqref{eq:ovalnegativek} is nonempty.
Let $X_0$ be of the form $X_0= \lambda P$. We will find $\lambda>0$ so that $X_0$ is on the surface defined by \eqref{eq:ovalnegativek}. 
This happens if
$\lambda\, |P| + \kappa \,|\lambda - 1|\,|P| = b.$
Suppose $b> |P|$. Then we must have $\lambda> 1$. So
$b =  \lambda \,|P| + \kappa \,(\lambda -1)\,|P|$, and 
solving for $\lambda$ gives
$$\lambda = \frac{b + \kappa |P|}{(1+\kappa)\,|P| }.$$
Now suppose  $\kappa |P| \leq b \leq |P|$, and
let
$$\lambda = \frac{b - \kappa |P|}{(1-\kappa)\,|P| }.$$
We then have $0<\lambda\leq 1$ and $X_0=\lambda \,P$ is on the surface \eqref{eq:ovalnegativek}.
\footnote{Notice that the set $\{X:|X| + \kappa |X - P| \leq \kappa|P|+\delta\}$, $\delta>0$ is contained in the
the ball with center zero and radius $\delta/(1+\kappa)$.  
}

{\bf Case $\kappa=-1$.}
In this case to have a nonempty surface we show that $b$ must be bounded above and below.
If the surface is nonempty then 
by application of the triangle inequality we get that
\begin{equation}\label{eq:boundsforbwhenkappaequals-1}
-|P|\leq b\leq |P|.
\end{equation}
Viceversa, if \eqref{eq:boundsforbwhenkappaequals-1} holds, then the surface is nonempty.
In fact, the point $X_0=\lambda P$ with $\lambda=\dfrac{b + |P|}{2\, |P|}$ belongs to the surface.


{\bf Case $\kappa<-1$.} 
We will show that the surface \eqref{eq:ovalnegativek} is 
nonempty if and only if   
$b$ is bounded above.
First notice that by the triangle inequality and the fact that $\kappa<-1$ we have that $|X|+\kappa\,|X-P|\leq |P|$ for all $X$, and therefore if \eqref{eq:ovalnegativek} is nonempty then
\begin{align}\label{upperboundforb}
b \leq |P|.
\end{align}
On the other hand, if \eqref{upperboundforb} holds then will find $X_0 = \lambda P\in \Gamma$ for some $\lambda > 0$. We want
$\lambda |P| + \kappa |P| |\lambda -1 | = b.$
If we let $\lambda=\dfrac{b + \kappa |P|}{|P| + \kappa |P|}$, then $\lambda \geq 1$ and the point $\lambda P$ is on the curve.
\footnote{Notice that the set $E=\{X:|X| + \kappa |X - P| \geq |P|-\delta\}$, $\delta>0$, is contained in the ball with center $0$ and radius $-\delta/(1+\kappa)$.
}

\subsection{Polar equation of \eqref{eq:ovalnegativek}}
Next we will find the polar equation of the refracting surface \eqref{eq:ovalnegativek} when $\kappa < 0$. We show that only a piece of this surface does the refracting job, otherwise total internal reflection occurs. This follows from the physical conditions described in Section \ref{subsubsec:physicalconditionforrefraction}.

Let $X = \rho(x)x$ with $x \in S^2$. Then we write
$$\kappa | \rho(x)x - P| = b - \rho(x).$$
Squaring both sides yields the following quadratic equation  
\begin{equation}\label{eq:quadraticeqforrho}
(1-\kappa^2)\,\rho(x)^2+2\,\left(\kappa^2\, x\cdot P -b\right)\,\rho(x)+b^2-\kappa^2\,|P|^2=0.
\end{equation}
Solving for $\rho$ we obtain
\begin{align*}
\rho(x) = \frac{(b - \kappa^2 x \cdot P) \pm \sqrt{(b - \kappa^2 x \cdot P)^2 - (1 - \kappa^2)(b^2 - \kappa^2 |P|^2)}}{1 - \kappa^2}.
\end{align*}
Define an auxiliary quantity
$$\Delta (t) := (b - \kappa^2 t)^2 - (1 - \kappa^2)(b^2 - \kappa^2 |P|^2)$$ so that
$$\rho_\pm(x) = \frac{(b - \kappa^2 x \cdot P) \pm \sqrt{\Delta (x \cdot P)}}{1 - \kappa^2}. $$

\subsubsection{Case $-1 < \kappa < 0$}
We begin with a proposition.
\begin{proposition}\label{prop:lowerestimateofDelta}
Assume $-1 < \kappa < 0$. If $x\in S^2$, then
\begin{align}
\Delta(x \cdot P) \geq  \kappa^2 (x \cdot P - b)^2,
\end{align}
with equality if and only if $|x\cdot P|=|P|$.
Hence, the quantity $\sqrt{\Delta ( x \cdot P)}$ is well defined. 
\end{proposition}

\begin{proof}
First notice that $| x \cdot P | \leq |P|$ comes for free since $|x| = 1$. We have, since $0 < \kappa^2 < 1$ and $x \cdot P \leq |P|$, 
\begin{dmath}
\Delta (x \cdot P) - \kappa^2 (x \cdot P - b)^2 = -2b\kappa^2 x \cdot P + \kappa^4 (x \cdot P)^2 + \kappa^2 b^2 - \kappa^4 |P|^2 + \kappa^2 |P|^2 - \kappa^2 (x \cdot P)^2 + 2 \kappa^2 b x \cdot P - \kappa^2 b^2 = k^4 (x \cdot P)^2 - \kappa^4 |P|^2 + \kappa^2 |P|^2 - \kappa^2 (x \cdot P)^2 = -\kappa^4 \left( |P|^2 - (x \cdot P)^2 \right) + \kappa^2 \left( |P|^2 - (x \cdot P)^2 \right) = \kappa^2 (1 - \kappa^2) \left( |P|^2 - (x \cdot P)^2 \right) \geq 0
\end{dmath}
as desired. 
\end{proof}

Now, notice that in this case, in order to make sense of the physical problem, we must have $P$ lying outside of the oval; that is, outside of the region $|X| + \kappa |X-P| \leq b$. In addition to the requirement that the surface be nonempty, we see that the condition on $b$ in order that the ovals are meaningful is that $$\kappa |P| < b < |P|$$

Recalling that
$$\kappa | \rho(x)x - P| = b - \rho(x)$$ we see that, since $ \kappa < 0$, we must have $\rho(x) > b$. We now choose which sign ($\pm$) to take in the definition of $\rho_\pm$ so that $\rho>b$. 
In fact, by Proposition \ref{prop:lowerestimateofDelta},
$$\rho_+ (x)=\frac{(b - \kappa^2 x \cdot P) + \sqrt{\Delta (x \cdot P)}}{1 - \kappa^2}\geq \frac{(b - \kappa^2 x \cdot P) + |\kappa| |b - x \cdot P|}{1 - \kappa^2} \geq b, $$
where the last inequality holds since $|\kappa|<1$.
Therefore the equality $\rho_+(x)=b$ holds only if $| x \cdot P| = |P|$ and $b = x \cdot P$. 
Also 
\begin{align*}
\frac{(b - \kappa^2 x \cdot P) + |\kappa| |b - x \cdot P|}{1 - \kappa^2}
&=
\dfrac{1}{1-\kappa^2}
\begin{cases}
(b - \kappa^2 x \cdot P) + |\kappa| (b - x \cdot P) & \text{for $b > x \cdot P$}\\
(b - \kappa^2 x \cdot P) + |\kappa| (x \cdot P-b) & \text{for $b < x \cdot P$}
\end{cases}\\
&\quad >b.
\end{align*}
Hence $\rho_+ (x) > b$ for all $x$ and $b$ such that $b\neq x\cdot P$.
From \eqref{bgeqkappap} the oval is not degenerate only for $\kappa|P| < b$. 
Reversing the above inequalities one can show that
$$\rho_- (x) \leq b.$$ 
Thus the polar equation of the surface $\Gamma$ is given by
\begin{align}
h(x, P, b) = \rho_+ (x) = \frac{(b - \kappa^2 x \cdot P) + \sqrt{\Delta (x \cdot P)}}{1 - \kappa^2}
\end{align}
provided that $\kappa|P| < b $. Furthermore, from the physical constraint for refraction \eqref{eq:constraint_kappa_lessthan0}, in this case with $m=\dfrac{P-\rho_+(x)x}{|P-\rho_+(x)x|}$, so $x\cdot m\geq \kappa$, and using the equation \eqref{eq:ovalnegativek}, we must have
$$x \cdot P \geq b.$$
Ending with the case when $-1 < \kappa < 0$, given a point $P \in \mathbb{R}^3$ and $\kappa |P| \leq  b $, a \emph{refracting oval} is the set
\begin{align} \label{refractingoval1}
\mathcal{O}(P, b) = \left\{ h(x, P, b)x : x \in S^2 \;, x \cdot P \geq b\right\} 
\end{align}
 where
$$h(x, P, b) = \frac{ (b - \kappa^2 x \cdot P) + \sqrt{(b - \kappa^2 x \cdot P)^2 - (1 - \kappa^2)(b^2 - \kappa^2 |P|^2)}}{1 - \kappa^2}.$$
Figure \ref{fig:ovalwithkappa=-.7}  illustrates an example of such a refracting oval.

\begin{figure}
\begin{center}
\includegraphics[scale=.4,angle=90]{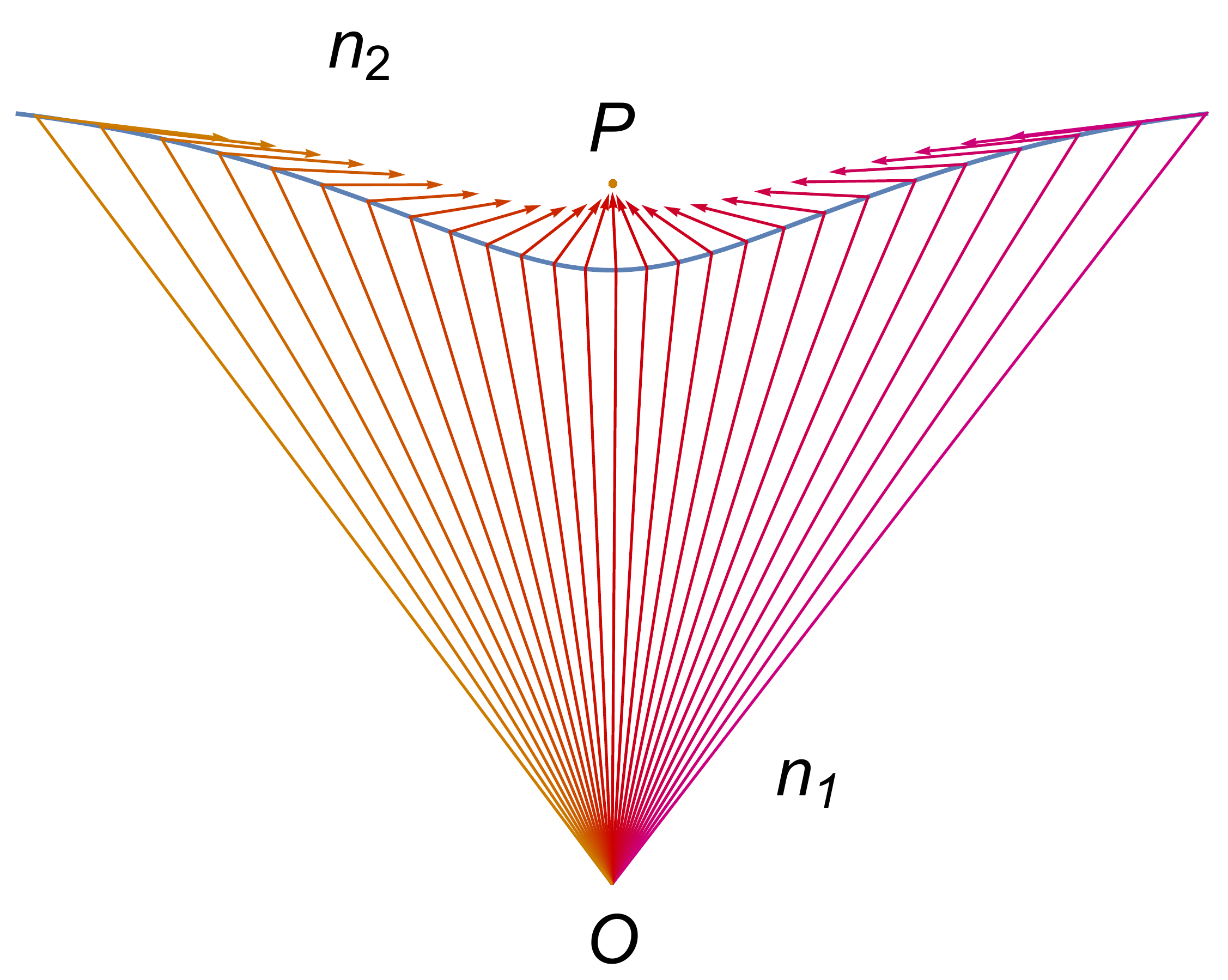}
\caption{$\kappa=-.7$, $P=(0,2.5)$.}
\label{fig:ovalwithkappa=-.7}
\end{center}
\end{figure}

\subsubsection{Case $\kappa < -1$.}


In this case we must have $O$ not on the oval, that is, $O$ must lie outside of the set $|X| + \kappa |X-P| \geq b$. Hence in combination with the requirement that the surface be nonempty we must have $$\kappa |P| < b < |P|.$$ Notice this is the same condition on $b$ as in the case $-1 < \kappa < 0$.

Multiplying \eqref{eq:quadraticeqforrho} by $-1$ and solving for $\rho_{\pm}$ yields
\[
\rho_{\pm}(x)=\dfrac{\kappa^2 \,x\cdot P -b \pm \sqrt{(b-\kappa^2\,x\cdot P)^2-(\kappa^2-1)(\kappa^2 |P|^2-b^2)}}{\kappa^2-1}.
\]
To have the square root defined we need 
\[
\Delta(x\cdot P)=(b-\kappa^2\,x\cdot P)^2-(\kappa^2-1)(\kappa^2 |P|^2-b^2)\geq 0.
\]
From \eqref{upperboundforb}, to have a nonempty oval we must have $b\leq |P|$.
To find the directions $x$ for which $\rho$ is well defined, we need to find the values of $t$ for which
\[
\Delta(t)=(b-\kappa^2\,t)^2-(\kappa^2-1)(\kappa^2 |P|^2-b^2)\geq 0.
\]
Let $z=b-\kappa^2\,t$, and  then we want $t$ such that
\[
z^2\geq (\kappa^2-1)(\kappa^2 |P|^2-b^2).
\]
If $\kappa^2 |P|^2-b^2\leq 0$, then this is true for any $t$.
On the other hand, $\kappa^2 |P|^2-b^2\geq 0$ if and only if $|b|\leq -\kappa |P|$.
So for $b$ with $|b|\leq -\kappa |P|$ we have to choose $t$ such that $z=b-\kappa^2\,t$ satisfies\footnote{Notice the requirement $|b| \leq -\kappa |P|$ is satisfied since we are assuming that $b \leq |P|$ so in particular $|b| \leq |P| \leq -\kappa |P|$.}
\[
|z|\geq \sqrt{(\kappa^2-1)(\kappa^2 |P|^2-b^2)}.
\] 
If $b-\kappa^2\,t\geq 0$, i.e., $t\leq b/\kappa^2$, then 
$
b-\kappa^2\,t\geq \sqrt{(\kappa^2-1)(\kappa^2 |P|^2-b^2)},
$
that is,
\[
t\leq \dfrac{b-\sqrt{(\kappa^2-1)(\kappa^2 |P|^2-b^2)}}{\kappa^2}.
\]
If $b-\kappa^2\,t\leq 0$, i.e., $t\geq b/\kappa^2$, then 
$
\kappa^2\,t-b\geq \sqrt{(\kappa^2-1)(\kappa^2 |P|^2-b^2)},
$
that is,
\begin{equation}\label{eq:tbiggerthatnb+sqrtoverkappasquared}
t\geq \dfrac{b+\sqrt{(\kappa^2-1)(\kappa^2 |P|^2-b^2)}}{\kappa^2}.
\end{equation}
From \eqref{eq:constraint_kappa_lessthan-1}, to have refraction in this case, we need to have
$$x \cdot \frac{P - x \rho_{\pm}(x)}{|P - x \rho_{\pm}(x)|} \geq \frac{1}{\kappa}.$$
Note that this is equivalent to requiring that
$$\kappa^2 x \cdot P - b \geq (\kappa^2 -1) \rho_{\pm}(x) $$ 
That is, we need
$$\rho_{\pm}(x) \leq \frac{\kappa^2 x \cdot P - b}{\kappa^2 -1}.$$
Note that since $\kappa^2 > 1$, we see that $\rho_{\pm}(x) < 0$ for $\kappa^2 x \cdot P - b < 0$. Thus we must have $\kappa^2 x \cdot P \geq b$ if we are to have $\rho_{\pm}(x) \geq 0$. From \eqref{eq:tbiggerthatnb+sqrtoverkappasquared} we need that
$$x \cdot P \geq \frac{b + \sqrt{ (\kappa^2 - 1)(\kappa^2 |P|^2 - b^2)}}{\kappa^2},$$
so that $\Delta(x\cdot P)\geq 0$.
We are left with choosing either $\rho_+ (x)$ or $\rho_- (x)$ now. 
But by the physical restraint for refraction, we see only $\rho_-(x)$ will do the job since
$$\rho_- (x) = \frac{ \kappa^2 x \cdot P - b - \sqrt{\Delta (x \cdot P)}}{\kappa^2 -1} \leq \frac{\kappa^2 x \cdot P - b}{\kappa^2 -1} $$ provided that
$$x \cdot P \geq \frac{ b + \sqrt{ (\kappa^2 -1)(\kappa^2 |P|^2 - b^2)}}{\kappa^2}.$$  
Hence for $\kappa < -1$, refraction only occurs when $b \leq |P|$ and the refracting piece of the oval is given by
\begin{align} \label{refractingoval2}
\mathcal{O}(P, b) = \left\{ h(x, P, b)x: x \cdot P \geq \frac{ b + \sqrt{ (\kappa^2 -1)(\kappa^2 |P|^2 - b^2)}}{\kappa^2} \right\}
\end{align}
with
$$h(x, P, b) = \rho_- (x) = \frac{(\kappa^2 x \cdot P - b) - \sqrt{(\kappa^2 x \cdot P - b)^2 - (\kappa^2 -1)(\kappa^2 |P|^2 - b^2)}}{\kappa^2 -1}.$$
A picture of the refracting piece of this oval can be obtained from Figure \ref{fig:ovalwithkappa=-.7} by reversing the roles of $P$ and $O$ 
and changing the direction of the rays; see also the explanation at the end of the proof of Theorem \ref{klessthan0}.

\subsubsection{Case $\kappa=-1$}
In this case we have $|\rho(x)x-P|=\rho(x)-b$, so squaring and solving for $\rho$ yields
\[
\rho(x)=\dfrac{b^2-|P|^2}{2\,\left(b-x\cdot P \right)}.
\]
To have a nonempty surface, $b$ must satisfy \eqref{eq:boundsforbwhenkappaequals-1}. 
Since $\rho\geq 0$ and the numerator in $\rho$ is nonpositive, we must have $b<x\cdot P$.
The last inequality follows from the physical condition for refraction \eqref{eq:constraint_kappa_lessthan0} with $\kappa=-1$ and $m=\dfrac{P-\rho(x)x}{|P-\rho(x)x|}$. 
Notice that $\rho(x)\geq b$, that is, $\dfrac{b^2-|P|^2}{2\,\left(b-x\cdot P \right)}\geq b$, since this is equivalent to $(b-x\cdot P)^2+|P|^2-(x\cdot P)^2\geq 0$.

Therefore when $\kappa=-1$ the surface refracting the origin $O$ into $P$ is given by
\begin{equation}\label{eq:refractingsurfacekappa-1}
\mathscr E(P,b)=\left\{\dfrac{b^2-|P|^2}{2\,\left(b-x\cdot P \right)}\,x:x\in S^2,\, x\cdot P>b \right\}
\end{equation}
with $|b|\leq |P|$. 

The uniform refraction property is summarized below. 

\begin{theorem} Let $n_1, n_2$ be two indices of refraction of two media $I$ and $II$, respectively, such that $\kappa := n_2/ n_1 < 0$. Assume that $O$ is a point inside medium $I$, $P$ is a point inside medium $II$, and  $b \in (\kappa |P|, |P|)$. Then,

\begin{enumerate}
\item If $-1 < \kappa < 0$ and $\Gamma := \mathscr{O}(P, b)$ is given by (\ref{refractingoval1}), 
then $\Gamma$ refracts all rays emitted from $O$ into $P$. 
\item If $\kappa < -1$ and $\Gamma: = \mathscr{O}(P, b)$ is given by (\ref{refractingoval2}), 
then $\Gamma$ refracts all rays emitted from $O$ into $P$.
\item If $\kappa = -1$ and $\Gamma: = \mathscr{E}(P, b)$ is given by (\ref{eq:refractingsurfacekappa-1}), 
then $\Gamma$ refracts all rays emitted from $O$ into $P$. 
\end{enumerate}

\end{theorem}

\subsection{Analysis of the convexity of the surfaces}\label{subsect:convexityofsurfaces}
We assume $\kappa<0$ and analyze the convexity of the 
surface
\[
|X|+\kappa |X-P|=b.
\]
We will show the surface is convex for a range of $b$'s and neither convex nor concave for the remaining range of $b$'s. \footnote{To be precise, by convexity we mean that the set that is enclosed by $F(X) = |X| + \kappa |X-P| = b$ is a convex set.} We begin by letting $X=(x_1,\cdots,x_n)$ and $P=(p_1,\cdots, p_n)$ and analyzing the function 
\[
F(X)=|X|+\kappa |X-P|.
\]
We have the gradient
\[
D F(X)=\dfrac{X}{|X|}+\kappa \,\dfrac{X-P}{|X-P|},
\]
and
\[
F_{x_ix_j}=\delta_{ij} \dfrac{1}{|X|} -\dfrac{x_ix_j}{|X|^3}+\kappa\,\delta_{ij} \dfrac{1}{|X-P|}
-
\kappa\, \dfrac{(x_i-p_i)(x_j-p_j)}{|X-P|^3},
\]
where $\delta_{ij}$ is the Kronecker delta. 
Then
\begin{align*}
&Q(X,\xi)=\sum_{i,j=1}^n F_{x_ix_j}(X)\xi_i\xi_j\\
&=
\dfrac{|\xi|^2}{|X|}
-
\dfrac{1}{|X|^3}\sum_{i,j=1}^n x_i x_j \xi_i\xi_j
+\kappa\,\dfrac{|\xi|^2}{|X-P|}
-
\kappa\,
\dfrac{1}{|X-P|^3}\sum_{i,j=1}^n(x_i-p_i)(x_j-p_j)\xi_i\xi_j\\
&=
\dfrac{|\xi|^2}{|X|}
-
\dfrac{1}{|X|^3}\,(X\cdot \xi)^2
+\kappa\,\dfrac{|\xi|^2}{|X-P|}
-
\kappa\,
\dfrac{1}{|X-P|^3}((X-P)\cdot \xi)^2\\
&=
\kappa\,
\dfrac{1}{|X-P|^3}\left[|\xi|^2\,|X-P|^2 -((X-P)\cdot \xi)^2\right]
+
\dfrac{1}{|X|^3}\,\left[|\xi|^2\,|X|^2- (X\cdot \xi)^2\right].
\end{align*}
By Cauchy-Schwartz's inequality $|v\cdot w|\leq |v|\,|w|$ with equality if and only if $v$ is a multiple of $w$.
If we set $\xi=X$ in the quadratic form we get $Q(X,X)=\kappa\,
\dfrac{1}{|X-P|^3}\left[|X|^2\,|X-P|^2 -((X-P)\cdot X)^2\right]$. Since $P\neq 0$, $X$ is not a multiple of $X-P$ if and only if $X$ is not a multiple of $P$ and in this case we get $Q(X,X)<0$ since $\kappa < 0$. On the other hand, if we set $\xi=X-P$, then
$Q(X,X-P)=\dfrac{1}{|X|^3}\,\left[|X-P|^2\,|X|^2- (X\cdot X-P)^2\right]$, and again if $X$ is not a multiple of $P$ we obtain $Q(X,X-P)>0$. Therefore the quadratic form $Q(X,\xi)$ is indefinite and therefore the function $F$ is neither concave nor convex at each $X$ which is not a multiple of $P$.

\begin{figure}[htp]
\begin{center}
    \subfigure[Ovals with $\kappa = -.5$, $P=(2,0)$, and different values of $-1\leq b\leq 4$.]{\label{fig:edge-a}
    \includegraphics[width=4in]{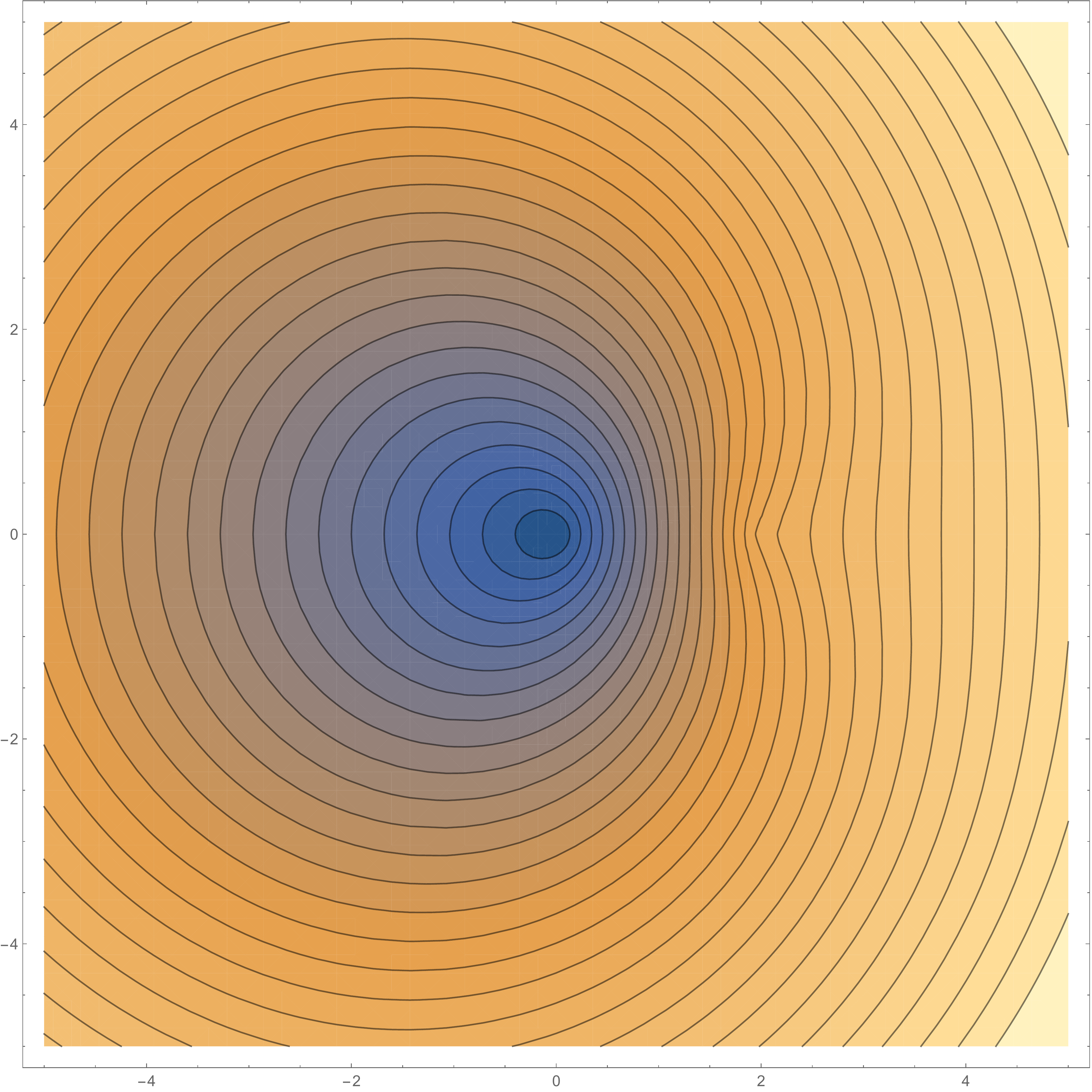}}
    \subfigure[scale]{\label{fig:edge-b}
    \includegraphics[width=.3in]{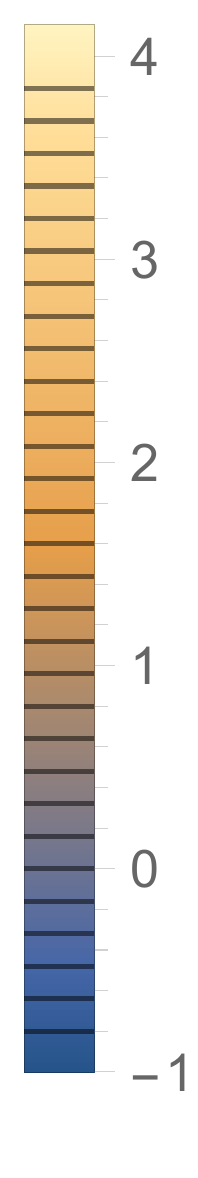}} 
\end{center}
  \caption{Ovals $\kappa=-.5$}
  \label{fig:ovalskappa=-.5}
\end{figure}

Let now $X$ have the form $X=\lambda \,P$, $0<\lambda<\infty$, 
and see for what values of $\lambda$ there holds that $Q(\lambda P,\xi)\geq 0$ for all $\xi$. In other words, we study which values of $\lambda$ make the Hessian $F_{x_ix_j}(\lambda P)\geq 0$. In particular {\it we wish to find the range of $b$ that give a convex surface. The range we obtain will be proven rigorously in Theorem \ref{klessthan0}.} 
First off we have
\begin{align*}
&Q(\lambda P,\xi)\\
&= 
\kappa\,\dfrac{1}{|1-\lambda|^3|P|^3}
\left[
(1-\lambda)^2 |\xi|^2|P|^2 -(1-\lambda)^2(P\cdot \xi)^2
\right]
+
\dfrac{1}{\lambda^3|P|^3 }\left[\lambda^2 |\xi|^2 |P|^2 - \lambda^2 (P\cdot \xi)^2 \right]\\
&=
\left(\dfrac{1}{\lambda}+\kappa\,\dfrac{1}{|1-\lambda|}\right)\,\dfrac{1}{|P|^3}\,\left[|\xi|^2 |P|^2 -  (P\cdot \xi)^2 \right]\geq 0
\end{align*}
if and only if
\[
\phi(\lambda):=\dfrac{1}{\lambda}+\kappa\,\dfrac{1}{|1-\lambda|}\geq 0.
\]
Let us assume $-1<\kappa<0$.
Then $\phi(\lambda)> 0$ on $(0,1/(1-\kappa))$, $\phi(\lambda)< 0$ on $(1/(1-\kappa),1)$,
$\phi(\lambda)< 0$ on $(1,1/(1+\kappa))$, and $\phi(\lambda)> 0$ on $(1/(1+\kappa),+\infty)$.
Therefore
\[
Q(\lambda P,\xi)
\begin{cases}
> 0 & \text{for $\lambda\in (0,1/(1-\kappa))\cup (1/(1+\kappa),+\infty)$}\\
< 0 & \text{for $\lambda\in (1/(1-\kappa),1)\cup (1,1/(1+\kappa))$} .
\end{cases}
\]
Suppose $b\geq \kappa |P|$, see (\ref{bgeqkappap}).
If $\lambda\in (0,1/(1-\kappa))$, then since $\kappa<0$, we have $\lambda <1$.
If the point $X=\lambda P$ is on the surface \eqref{eq:ovalnegativek}, then 
$\lambda |P|+\kappa \,(1-\lambda)|P|=b$ and so $\lambda=\dfrac{b-\kappa\,|P|}{(1-\kappa)|P|}$.
Therefore when $0<\lambda<1/(1-\kappa)$ to have $F_{x_ix_j}(\lambda P)\geq 0$, we must have
\[
\dfrac{b-\kappa\,|P|}{(1-\kappa)|P|}\leq \dfrac{1}{1-\kappa}.
\]
That is,
\[
b\leq (1+\kappa)|P|.
\]
So $F_{x_ix_j}(\lambda P)\geq 0$, $0<\lambda<1/(1-\kappa)$, with $\lambda P$ on the surface $\lambda |P|+\kappa \,(1-\lambda)|P|=b$ if and only if $b\in [\kappa |P|,(1+\kappa)|P|]$.

Suppose now $\lambda \in (1/(1- \kappa), 1)$. Let us see for which $b$ we have $F_{x_i x_j}(\lambda P) \leq 0$. Since $\lambda= \dfrac{b-\kappa\,|P|}{(1-\kappa)|P|}$ we see that in order to have $F_{x_i x_j }(\lambda P) \leq 0$, we need

$$\frac{1}{1 - \kappa} < \frac{b - \kappa |P|}{(1 - \kappa)|P|} < 1$$ so that
$$(1 + \kappa)|P| < b < |P|$$
That is, $F_{x_ix_j}(\lambda P)\leq 0$ precisely when $b\in [(1+\kappa)|P|,|P|]$. 

On the other hand, if $\lambda\in (1/(1+\kappa),+\infty)$, then $\lambda>1$ since $-1<\kappa<0$.
If $\lambda P$ is on the surface, which means $\lambda |P|+\kappa \,(\lambda-1)|P|=b$, we have
\[
\lambda=\dfrac{b+\kappa\,|P|}{(1+\kappa)\,|P|}.
\]
So $\lambda\geq \dfrac{1}{1+\kappa}$ implies
\[b\geq (1-\kappa)|P|.
\] 
But $(1-\kappa)|P| > |P|$ and by the physical restriction on $b$, we must have $\kappa |P| <b < |P|$. So we see that 
$$|P| < (1-\kappa)|P| \leq b < |P|$$ which is impossible.

Let now $\lambda \in (1, 1/(1 + \kappa))$. Then since $\lambda =\dfrac{b+\kappa\,|P|}{(1+\kappa)\,|P|}$ we see that
$$ 1 < \dfrac{b+\kappa\,|P|}{(1+\kappa)\,|P|} < \frac{1}{1 + \kappa}$$
so that

$$|P| < b < (1 - \kappa)|P|$$
But by the constraint on $b$ we must have $b < |P|$, which contradicts the above inequality.

Hence given the critical values of $b$ as discussed above, the main result in this section is as follows: 

\begin{theorem} \label{klessthan0}
Suppose $-1<\kappa <0$. 
The surface 
$|X| +\kappa\, |X-P| = b$
is
\begin{enumerate}
\item convex, when $b\in [\kappa |P|,(1+\kappa)|P|]$;
\item neither convex nor concave, when $b\in ((1+\kappa)|P|,|P|)$.
\end{enumerate}

If $\kappa < -1$ then the surface 
$|X| +\kappa\, |X-P| = b$
is
\begin{enumerate}
\item convex when $b \in ((1+ \kappa)|P|, |P|)$;
\item neither concave nor convex when $b \in (\kappa |P|, (1+\kappa)|P|]$.
\end{enumerate}
\end{theorem}


\begin{proof}
We begin by proving Theorem \ref{klessthan0} in dimension two when $-1 < \kappa < 0$. 
We will determine the curvature using \cite[Proposition 3.1]{Goldman}. That is, we let $\xi = (-F_y, F_x)$ into the quadratic form $Q(X, \xi)$ and determine its sign for each $X$ on the curve when $b$ varies. 
We have that 
\begin{align*}
&Q(X, (-F_y, F_x))  = \frac{1}{|X-P|^3 |X|^3} \times \\ 
&\left[ \kappa^3 |X|^3 |X-P|^2 + |X|^2 |X-P|^3 \right.
\\& \left.+ \kappa^2 |X-P|\left( X \cdot P -3 |X|^2 \right) \left( X \cdot P - |X|^2 \right) 
+  \kappa |X|  \left( |X|^2 - X \cdot P \right) \left( 2 |P|^2 -5 X \cdot P + 3 |X|^2 \right)  \right]\\
& := \frac{1}{|X-P|^3 |X|^3} \times 
\left[ \kappa^3 |X|^3 |X-P|^2 + |X|^2 |X-P|^3 
+ B + 
 C  \right].
\end{align*}
Since $2( X \cdot P ) = |X|^2 + |P|^2 - |X-P|^2$, we can write $B,C$ as functions of $|X|$ and $|X-P|$ as follows:
\begin{align*}
B&=
\kappa^2 |X-P|\left(  3 |X|^2 -X \cdot P\right) \left(|X|^2- X \cdot P  \right)\\
&=
\kappa^2 |X-P| \left(|X|^2- X \cdot P  \right)^2
+2\kappa^2 |X-P|\, |X|^2\, \left(|X|^2- X \cdot P  \right)\\
&=
\dfrac14 \,\kappa^2 |X-P| \left(|X|^2+|X-P|^2-|P|^2  \right)^2
+\kappa^2 |X-P|\, |X|^2\, \left(|X|^2+|X-P|^2-|P|^2  \right);
\end{align*}
and
\begin{align*}
C&=
\kappa |X|  \left( |X|^2 - X \cdot P \right) \left( 2 |X-P|^2+|X|^2-X\cdot P \right)\\
&=
2\,\kappa |X| \,|X-P|^2 \left( |X|^2 - X \cdot P \right) 
+
\kappa |X| \left( |X|^2 - X \cdot P \right)^2\\
&=
\kappa |X| \,|X-P|^2 \left( |X|^2+|X-P|^2-|P|^2 \right) 
+
\dfrac14\,\kappa |X| \left( |X|^2+|X-P|^2-|P|^2 \right)^2.
\end{align*}
The numerator of $Q(X, (-F_y, F_x))$ then equals
\begin{align*}
&
|X|^2 |X-P|^2 \left(\kappa^3 |X|+|X-P| \right)
+
\dfrac14 \kappa \left(|X-P|^2+|X|^2-|P|^2\right)^2 \left(|X|+\kappa |X-P|\right)\\
&\qquad 
+ \kappa |X| |X-P| \left(|X-P|^2+|X|^2-|P|^2\right)\left(\kappa |X|+|X-P| \right):=H(|X|,|X-P|,\kappa,|P|),
\end{align*} 
and we want to optimize this quantity when $X$ is on the curve $F(X)=|X|+\kappa |X-P|=b$.
If $X$ is on the curve there holds $|X-P| = \dfrac{b-|X|}{\kappa}\geq 0$.
Hence $|X|\geq b$, and also $\left||P|-|X|\right|\leq |X-P|=\dfrac{b-|X|}{\kappa}$.
Hence
\[
-\dfrac{b-|X|}{\kappa}\leq |P|-|X|\leq \dfrac{b-|X|}{\kappa}.
\]
Multiplying by $\kappa<0$ yields
\[
-b+|X|\geq \kappa |P|-\kappa |X|\geq b-|X|,
\]
so
\[
(1-\kappa) |X|\geq b-\kappa |P|,
\]
and
\[
(1+\kappa) |X|\geq b+\kappa |P|.
\]
Then the last two inequalities together with $|X|\geq b$ imply that $X$ must satisfy
\[
|X|\geq \max \left\{\dfrac{b-\kappa |P|}{1-\kappa}, \dfrac{b+\kappa |P|}{1+\kappa},b\right\}.
\]
In addition, from $|X-P| = \dfrac{b-|X|}{\kappa}$ we get that $|X|+|P|\geq |X-P| = \dfrac{b-|X|}{\kappa}$, which implies
\[
|X|\leq \dfrac{b-\kappa |P|}{1+\kappa}\footnote{This means that the ball centered at zero with radius $\dfrac{b-\kappa |P|}{1+\kappa}$ contains the oval. In fact, this is the smallest ball centered at zero that contains the oval because the point $X=\lambda P$ with $\lambda=-\dfrac{b-\kappa |P|}{(1+\kappa)|P|}$ satisfies the equation $F(X)=b$ (notice that $\lambda\leq 0$ and $\lambda-1=\dfrac{-b-|P|}{(1+\kappa)|P|}< 0$ since $-b\leq -\kappa |P|<|P|$ since $-1<\kappa<0$).}
\]
since $-1<\kappa<0$.
Substituting $|X-P| = \dfrac{b-|X|}{\kappa}$ in the expression for $H$ above we obtain the following expression in terms of $\kappa, b,|X|$ and $|P|$:
\begin{align*}
&\dfrac{b(b^2 - \kappa^2 |P|^2)^2 + 6b(\kappa^2 -1)(b - \kappa |P|)(b + \kappa |P|) |X|^2}{4 \kappa^3} \\
&\quad + \dfrac{ 4 (\kappa^2 -1)\left(b^2 \left( \kappa^2 -2 \right) + \kappa^2 |P|^2 \right) |X|^3 - 3b (\kappa^2 -1)^2 |X|^4}{4 \kappa^3}:=G(|X|,b,\kappa),
\end{align*}
with $|X|\geq b$.
Thus the problem of finding the $b$'s for which $Q(X, (-F_y, F_x))\geq 0$ for all $X$, reduces to minimizing a polynomial of one variable $|X|$ over the range $\max \left\{\dfrac{b-\kappa |P|}{1-\kappa}, \dfrac{b+\kappa |P|}{1+\kappa},b\right\}\leq |X|\leq \dfrac{b-\kappa |P|}{1+\kappa}$. Notice that the value of the maximum is given by $M = \dfrac{b- \kappa |P|}{1-\kappa}$\footnote{Here we use that $b\leq |P|$, because the point $P$ must be outside the oval to have a problem physically meaningful. This means that the radius of the largest ball centered at zero and contained in the oval has radius $\dfrac{b- \kappa |P|}{1-\kappa}$.}. 
The derivative of $G$ with respect to the first variable $t=|X|$ equals
\[
\partial_t G(t,b,\kappa)=\dfrac{3 (-1 + \kappa^2) (b - t) t (b^2 - \kappa^2 |P|^2 + b (-1 + \kappa^2) t)}
{k^3},
\] 
and therefore if $b\neq 0$ then $\partial_t G(t,b,\kappa)=0$ when $t=0,b,$ or $t_b:=\dfrac{-b^2 + \kappa^2 |P|^2}{b (-1 + \kappa^2)}$. If $b=0$ then $\partial_t G(t,0,\kappa)=0$ when $t=0$.
Furthermore, the coefficient of $t^4$ in the original function $G(t,b,k)$ equals $L_b:=-\dfrac{3\,b}{4\, \kappa^3}(1-\kappa^2)^2$.
%
%

To optimize $G(t,b,\kappa)$ we proceed by analyzing the following cases.
\vskip 0.2in

{\bf Case $\kappa |P|\leq b<0$.}

Since $\kappa<0$, then the coefficient $L_b< 0$. We also have $t_b>0$. 
We write
\[
\partial_t G(t,b,\kappa)=-3\,b\,\dfrac{(\kappa^2-1)^2}{\kappa^3}\,t\,(t-b)\,(t-t_b).
\] 
For $t_b<t<\infty$ we have $\partial_t G(t,b,\kappa)<0$, and for $0<t<t_b$ we have $\partial_t G(t,b,\kappa)>0$.
We optimize the function $G$ on the interval
$$I = \left(\dfrac{b - \kappa |P|}{1- \kappa}, \dfrac{b-\kappa |P|}{1+\kappa}\right) \subset (0,t_b).$$

We have that $\partial_t G(t,b,\kappa)>0$ for $t\in I$ and so $G$ is strictly increasing on $I$. 
Hence we obtain
$$\max_{t \in I} G(t,b,\kappa) = G \left(\dfrac{b-\kappa |P|}{1+\kappa}, b , \kappa \right) 
= \dfrac{(b + |P|)^2 (b - \kappa |P|)^2 (b + |P| - \kappa |P|)}{(1 + k)^2},$$
and 
$$\min_{t \in I} G(t, b, \kappa) = G\left(\dfrac{b-\kappa |P|}{1-\kappa}, b , \kappa \right)
= \dfrac{(b - |P|)^2 (b - \kappa |P|)^2 (-b + (1+\kappa)|P| )}{(-1 + k)^2}.$$

Note that this minimum is positive since $-1< \kappa < 0$ and the maximum is positive as well since $b \geq \kappa |P|$. Hence the maximum and minimum have the same sign and as such the curve is convex for this range of $b$.

\vskip 0.2in
{\bf Case $0< b \leq (1+\kappa)|P|$.}

By calculation we have that $L_b>0$.
Let us analyze the sign of $t_b$. Since $b>0$ and $-1<\kappa<0$, we have $t_b>0$ if and only if $b>-\kappa|P|$, and $t_b<0$ if and only if $b<-\kappa|P|$. Let us compare $-\kappa |P|$ with $(1+\kappa)|P|$. We have $-\kappa |P|\leq (1+\kappa)|P|$ if and only if $\kappa\leq -1/2$, and $-\kappa |P|> (1+\kappa)|P|$ if and only if $\kappa> -1/2$.
Therefore, if $\kappa>-1/2$ then $t_b<0$. On the other hand, if $\kappa\leq -1/2$, then $b\in (0,-\kappa |P|)\cup [-\kappa |P|, (1+\kappa)|P|)$. If $b\in (0,-\kappa |P|)$, then $t_b<0$. If $b\in [-\kappa |P|, (1+\kappa)|P|)$, then $t_b>0$.
Notice also that $t_b<b$ since $b<(1+\kappa)|P|$.
If $\kappa>-1/2$ then $t_b<0$ and $\partial_tG<0$ on $(0,b)$ while $\partial_tG>0$ on $(b,+\infty)$. 
Since $\dfrac{b-\kappa|P|}{1-\kappa}>b$, we have that the interval $I\subset (b,+\infty )$, and 
$$\min_{t \in I} G(t, b, \kappa) = G\left(\dfrac{b-\kappa |P|}{1-\kappa}, b , \kappa \right)
= \dfrac{(b - |P|)^2 (b - \kappa |P|)^2 (-b + (1+\kappa)|P| )}{(-1 + k)^2}.$$

If $\kappa\leq -1/2$ then $0<t_b<b$ and so $\partial_tG>0$ on $(0,t_b)$, $\partial_tG<0$ on $(t_b,b)$ and $\partial_tG>0$ on $(b,+\infty)$. Once again $I\subset (b,+\infty )$, and we obtain
$$\min_{t \in I} G(t, b, \kappa) = G\left(\dfrac{b-\kappa |P|}{1-\kappa}, b , \kappa \right)
= \dfrac{(b - |P|)^2 (b - \kappa |P|)^2 (-b + (1+\kappa)|P| )}{(-1 + k)^2}.$$

Note that this minimum is positive since $b < 0$. Furthermore we have
$$\max_{t \in I} G(t,b, \kappa) = G \left( \dfrac{b-\kappa |P|}{1+\kappa}, b,\kappa \right)$$
which is positive because $b \geq \kappa |P|$. Hence the maximum and minimum have the same sign and for this range of $b$'s the curve is also convex. 
\vspace{.2 in}


{\bf Case $(1+ \kappa)|P| < b < |P|$.}

In this case we have that $t_b < b < \dfrac{b-\kappa |P|}{1-\kappa}, \dfrac{b-\kappa |P|}{1+ \kappa}$ so that the interval $I = \left( \dfrac{b- \kappa |P|}{1-\kappa}, \dfrac{b-\kappa |P|}{1+\kappa} \right) \subset (b, \infty)$. Furthermore we have that $t_b > 0$ provided that $-1 < \kappa < -1/2$ and $-\kappa |P| < b< |P|$ or if $-1/2 < \kappa < 0$ and $(1+\kappa)|P| < b < |P|$. Moreover we have that $t_b < 0$ if $-1<\kappa < -1/2$ and $(1+\kappa)|P| < b < -\kappa |P|$. 

Also on the interval $I$ we have that $\partial_t G(t,b,\kappa) >0$ so that

$$\min_{t \in I} G(t, b, \kappa) = G \left( \dfrac{b- \kappa |P|}{1-\kappa}, b \kappa \right) = \dfrac{(b- |P|)^2(b- \kappa |P|)^2(-b +(1+\kappa)|P|)}{(-1+\kappa)^2} $$
which is negative since $b > (1+ \kappa)|P|$. Furthermore we have that

$$\max_{t \in I} G(t,b,\kappa) = G \left( \dfrac{b- \kappa |P|}{1+\kappa}, b, \kappa \right) = \dfrac{(b+ |P|)^2(b- \kappa |P|)^2(b + (1-\kappa)|P|)}{(1 +\kappa)^2} $$ which is positive since $-1<\kappa < 0$. Thus the minimum and maximum of $G$ on the interval have opposite signs which implies that the curve is neither convex nor concave for this range of $b$.

In dimension three, the above result still holds true because the oval is radially symmetric with respect to the axis $OP$. 
\vskip 0.2in
Finally suppose that $\kappa<-1$. 

We have $X\in \{X:|X|+\kappa|X-P|=b\}:=O$ if and only if $X-P\in \left\{Z:|Z|+\dfrac{1}{\kappa}|Z+P|=\dfrac{b}{\kappa}
\right\}:=O'$. The ovals $O$ and $O'$ have the same curvature, and 
by the $-1<\kappa<0$ case, $O'$ is convex provided $\dfrac{1}{\kappa}|P| < \dfrac{b}{\kappa} < \left(1+ \dfrac{1}{\kappa}\right)|P|$ and neither convex nor concave provided $\left(1+ \dfrac{1}{\kappa}\right)|P| < \dfrac{b}{\kappa} < |P|$. Multiplying both inequalities by $\kappa$ we obtain the desired range of $b$.

\end{proof}

\section{Fresnel Formulas for NIMs}\label{sec:maxwellandfresnelformulas}

%
To obtain the Fresnel formulas for NIMs materials we briefly review the calculations leading to the Fresnel formulas for standard materials, see  \cite[Section 1.5.2]{BW1959}.

The electric field is denoted by $\E$ and the magnetic field by $\B$. 
 These are three-dimensional vector fields:
 $\E=\E(\r, t)$ and $\B=\B(\r, t)$, where 
 $\r$ represents a point in 3 dimensional space $\r =(x,y,z)$.
 The way in which $\E$ and $\B$ interact is described by Maxwell's equations \cite{BW1959} and \cite[Sec. 4.8]{sommerfeld:electrodynamics}: 
\begin{align}
\nabla \times \E&=-\dfrac{\mu}{c}\dfrac{\partial \B}{\partial t},\label{faradaylawfree}\\
\nabla \times \B&=  \dfrac{2\pi}{c}\sigma \E+\dfrac{\epsilon}{c} \dfrac{\partial \E}{\partial t}\label{amperemaxwelllawfree}\\
\nabla \cdot (\epsilon \E)&=4\pi \rho\label{gausslawbis}\\
\nabla \cdot (\mu \B)&=0\label{magneticlawbis},
\end{align}
$c$ being the speed of light in vaccum. 

We consider plane waves solutions to the Maxwell equations \eqref{faradaylawfree}, \eqref{amperemaxwelllawfree}, \eqref{gausslawbis},
\eqref{magneticlawbis} with $\rho=0$ and $\sigma=0$, and having components of the form
\[
a \cos \left( \omega \left( t - \dfrac{\r \cdot {\bf s}}{v}\right)+\delta \right)
=
a \cos \left( \omega t - {\bf k}\cdot \r +\delta \right)
\]
with ${\bf k}=\dfrac{\omega}{v}{\bf s}$ and $a, \delta$ real numbers, and ${\bf s}$ is a unit vector.
If $\E(\r,t)=E({\bf k}\cdot \r-\omega t)$ and $\B(\r,t)=H({\bf k}\cdot \r-\omega t)$
solve the Maxwell equations \eqref{faradaylawfree}, \eqref{amperemaxwelllawfree}, \eqref{gausslawbis}, one can show that
\begin{equation}\label{eq:BequalskcrossE}
\B=\dfrac{c}{\mu \omega}({\bf k}\times \E), \text{ and }, \E=-\dfrac{c}{\epsilon \omega}({\bf k}\times \B).
\end{equation}

We assume that the incident vector ${\bf s}={\bf s^i}$ with 
\[
{\bf s}^i=\sin \theta_i \hat{\bf i}+ \cos \theta_i \hat{\bf k}.
\]
That is, ${\bf s}^i$ lives on the $xz$-plane and so the direction of propagation is perpendicular to the $y$-axis.
Also the boundary between the two media is the $xy$-plane, $\nu$ denotes the normal vector to the boundary at the point $P$, that is, $\nu$ is on the $z$-axis, and  
$\theta_i$ is the angle between the normal vector $\nu$ and the
incident direction ${\bf s}^i$ (as usual $\hat{\bf i},\hat{\bf j}, \hat{\bf k}$ denote the unit coordinate vectors). Let $I_\parallel$ and $I_\perp$ denote the parallel and perpendicular components, respectively, of the incident field. The electric field corresponding to this incident field is
\begin{equation}\label{eq:incidentfieldkappa<0}
\E^i(\r,t)=\left(-I_{\parallel} \cos \theta_i, I_{\perp}, I_{\parallel}\sin \theta_i\right)\, \cos \left( \omega \left( t - \dfrac{\r \cdot {\bf s}^i}{v_1}\right) \right)
=\E_0^i\, \cos \left( \omega \left( t - \dfrac{\r \cdot {\bf s}^i}{v_1}\right) \right),
\end{equation}
with 
\[
v_1=\dfrac{c}{\sqrt{\epsilon_1 \mu_1}}.
\]
From \eqref{eq:BequalskcrossE}, the magnetic field is then
\begin{align*}
\B^i(\r,t)&=\sqrt{\dfrac{\epsilon_1}{\mu_1}}({\bf s}^i\times \E^i)\\
&=\sqrt{\dfrac{\epsilon_1}{\mu_1}}\left(-I_{\perp} \cos \theta_i, -I_{\parallel}, I_{\perp}\sin \theta_i\right)\, \cos \left( \omega \left( t - \dfrac{\r \cdot {\bf s}^i}{v_1}\right) \right)
=
\B_0^i\, \cos \left( \omega \left( t - \dfrac{\r \cdot {\bf s}^i}{v_1}\right) \right).
\end{align*}

Let us now introduce ${\bf s}^t$, the direction of propagation of the transmitted wave, and let $\theta_t$ be the angle between the normal $\nu$ and ${\bf s}^t$. Similarly, ${\bf s}^r$ is the direction of propagation of the reflected wave and $\theta_r$ is the angle between the normal $\nu$ and ${\bf s}^r$. We have that 
${\bf s}^r=\sin \theta_r \hat{\bf i}+ \cos \theta_r \hat{\bf k}=\sin \theta_i \hat{\bf i}- \cos \theta_i \hat{\bf k}$.
Then the corresponding electric and magnetic fields corresponding to transmission are
\begin{align*}
\E^t(\r,t)&=\left(-T_{\parallel} \cos \theta_t, T_{\perp}, T_{\parallel}\sin \theta_t\right)\, \cos \left( \omega \left( t - \dfrac{\r \cdot {\bf s}^t}{v_2}\right) \right) =
\E_0^t\, \cos \left( \omega \left( t - \dfrac{\r \cdot {\bf s}^t}{v_2}\right) \right) \\
\B^t(\r,t)&=
\sqrt{\dfrac{\epsilon_2}{\mu_2}}\left(-T_{\perp} \cos \theta_t, -T_{\parallel}, T_{\perp}\sin \theta_t\right)\, \cos \left( \omega \left( t - \dfrac{\r \cdot {\bf s}^t}{v_2}\right) \right)
=
\B_0^t\, \cos \left( \omega \left( t - \dfrac{\r \cdot {\bf s}^t}{v_2}\right) \right);
\end{align*}
with 
\[
v_2=\dfrac{c}{\sqrt{\epsilon_2 \mu_2}}.
\]
There are similar formulas for the fields $\E^r(\r,t), \B^r(\r,t)$ corresponding to reflection. 
Since the tangential components of the electric and magnetic fields are continuous across the boundary, see e.g. \cite[Section 1.1.3]{BW1959},
%
we obtain the equations
 \[
 I_{\perp}+R_{\perp}=T_\perp,
 \qquad
 \cos \theta_i (I_\parallel - R_\parallel )=\cos \theta_t T_\parallel;
 \]
and 
 \[
 \dfrac{I_{\parallel}}{\sqrt{\dfrac{\epsilon_1}{\mu_1}}}+\dfrac{R_{\parallel}}{\sqrt{\dfrac{\epsilon_1}{\mu_1}}}=\dfrac{T_\parallel}{\sqrt{\dfrac{\epsilon_2}{\mu_2}}},
 \qquad
 \cos \theta_i \left(\dfrac{I_\perp}{\sqrt{\dfrac{\epsilon_1}{\mu_1}}} - \dfrac{R_\perp}{\sqrt{\dfrac{\epsilon_1}{\mu_1}}} \right)=\cos \theta_t \dfrac{T_\perp}{\sqrt{\dfrac{\epsilon_2}{\mu_2}}}.
 \]
 The wave impedance of the medium is defined by
 \[
 z=\sqrt{\dfrac{\mu}{\epsilon}},
 \]
 so setting $z_j=\sqrt{\dfrac{\mu_j}{\epsilon_j}} $, $j=1,2$,
 and solving the last two sets of linear equations yields
 \begin{align}\label{eq:fresnelformulasclassical1}
 T_\parallel &=\dfrac{2 \,z_1 \cos \theta_i}{z_2 \cos \theta_i + z_1 \cos \theta_t}\,I_\parallel\qquad 
 T_\perp =\dfrac{2\, z_1 \cos \theta_i}{z_1 \cos \theta_i + z_2 \cos \theta_t}\,I_\perp\\
 R_\parallel &=\dfrac{ z_2 \cos \theta_i-z_1\cos \theta_t}{z_2 \cos \theta_i + z_1 \cos \theta_t}\,I_\parallel\qquad
 R_\perp =\dfrac{ z_1 \cos \theta_i-z_2\cos \theta_t}{z_1 \cos \theta_i + z_2 \cos \theta_t}\,I_\perp.\label{eq:fresnelformulasclassical2} 
 \end{align}
 These are the {\it Fresnel equations} expressing the amplitudes of the reflected and transmitted waves in terms of the amplitude of the incident wave.

We now apply this calculation to deal with NIMS.
We will replace ${\bf s^i}$ by $x$ and ${\bf s^t}$ by $m$, and we also set 
\[
\kappa=-\dfrac{\sqrt{\epsilon_2\mu_2}}{\sqrt{\epsilon_1\mu_1}}.
\] 
Recall $\nu$ is the normal
to the interface.
We have $\cos \theta_i=x\cdot \nu$ and $\cos \theta_t=m\cdot \nu$.
Suppose medium I has $\epsilon_1>0$ and $\mu_1>0$; and medium II
has $\epsilon_2<0$ and $\mu_2<0$. In other words, medium I is "right-handed"
and medium II is "left-handed". This means in medium II the refracted ray $m$ stays on the same side 
of the incident ray $x$.

From the Snell law formulated in \eqref{eq:snelllawnegative} (notice that $\kappa<0$), we have
$
x-\kappa m=\lambda \nu,
$
so the Fresnel equations \eqref{eq:fresnelformulasclassical1} and \eqref{eq:fresnelformulasclassical2} take the form
 \begin{align}\label{eq:fresnelformulaourwriting}
 T_\parallel &=\dfrac{2 \, z_1\,x\cdot \nu}{(z_2 x+ z_1 m)\cdot \nu}\,I_\parallel
 =
\dfrac{2 \, z_1\,x\cdot (x-\kappa \,m)}{(z_2 x+ z_1 m)\cdot (x-\kappa \,m) }\,I_\parallel \\
 T_\perp &
 =\dfrac{2 \, z_1\, x\cdot \nu}{(z_1 x+ z_2 m)\cdot \nu}\,I_\perp
=\dfrac{2 \, z_1\,x\cdot (x-\kappa \,m)}{ (z_1 x+ z_2 m)\cdot (x-\kappa \,m)}\,I_\perp \notag
 \\
 R_\parallel &=\dfrac{ (z_2 x- z_1 m)\cdot \nu}{ (z_2 x+ z_1 m)\cdot \nu}\,I_\parallel
 =
 \dfrac{ (z_2 x- z_1 m)\cdot (x-\kappa \,m)}{ (z_2 x+ z_1 m)\cdot (x-\kappa \,m)}\,I_\parallel \notag\\
 R_\perp &=\dfrac{ (z_1 x- z_2 m)\cdot \nu}{(z_1 x+z_2 m)\cdot \nu}\,I_\perp
 =
  \dfrac{ (z_1 x- z_2 m)\cdot (x-\kappa \,m)}{(z_1 x+z_2 m)\cdot (x-\kappa \,m)}\,I_\perp. \notag
 \end{align}
Notice that the denominators of the perpendicular components are the same and likewise for the parallel components.

As an example suppose that $\epsilon_1, \mu_1$ are both positive and $\epsilon_2=-\epsilon_1$ and $\mu_2=-\mu_1$.
This is the so called mirror like material. Then $z_1=z_2$ and $\kappa=-1$, and so
$R_\parallel=R_\perp=0$ for all incident rays. This means that all the energy is transmitted and nothing internally reflected. Notice also that if $\kappa \neq  -1$ then there is a reflected wave, that is, $R_\perp$ and $R_\parallel$ may be different from zero. 

 \begin{remark}\rm
 The wave impedance, like the index of refraction, is a unique characteristic of the medium in consideration. However, unlike the refractive index $n$, the wave impedance $z$ remains \emph{ positive} for negative values of $\epsilon_i, \mu_i$. Also, compare equations \eqref{eq:fresnelformulasclassical2} with the table in \cite[pg. 765]{Ves2003} giving $R_\perp$. There the Fresnel formula for $r_\perp$ in the exact Fresnel formula column contains a misprint. We believe our calculations above to be correct and consistent with what is expected in the case of negative refraction. 
 \end{remark}

\subsection{The Brewster angle}
This is the case when $R_\parallel=0$, this means when 
\[
z_2 \cos \theta_i-z_1\cos \theta_t=0,
\]
which together with the Snell law ($\sin \theta_i=\kappa \,\sin \theta_t$) yields
\[
\left( \dfrac{z_2}{z_1}\right)^2\,\cos^2 \theta_i +\dfrac{1}{\kappa^2}\, \sin^2\theta_i=1=
\cos^2\theta_i +\sin^2 \theta_i.
\]
So
\[
\tan^2 \theta_i=\dfrac{\left( \dfrac{z_2}{z_1}\right)^2-1}{1-\dfrac{1}{\kappa^2}}
=\dfrac{\mu_2}{\mu_1}\left(\dfrac{\epsilon_1\mu_2-\epsilon_2\mu_1}{\epsilon_2\mu_2-\epsilon_1\mu_1}\right)
\]
if $1/\kappa^2\neq 1$, and the Brewster angle is $\theta_i$ with
\[
\tan \theta_i=\sqrt{\dfrac{\mu_2}{\mu_1}\left(\dfrac{\epsilon_1\mu_2-\epsilon_2\mu_1}{\epsilon_2\mu_2-\epsilon_1\mu_1}\right)}.
\]
 
Let us compare this value of $\theta_i$ with the case when $\kappa > 0$. As above the Brewster angle occurs when $R_\parallel = 0$. In this situation, the reflected and transmitted rays are orthogonal to each other, and via Snell's law, it follows that
 $$\tan \theta_i = n_2/ n_1$$
 which is consistent with what we would expect. 

\subsection{Fresnel coefficients for NIMs}
To calculate these coefficients we follow the calculations from \cite[Section 1.5.3]{BW1959}.
The Poynting vector is given by
$
\S=\dfrac{c}{4\pi}\E\times \B,
$
where $c$ is the speed of light in free space.
From \eqref{eq:BequalskcrossE} we get that 
\[
\S=\dfrac{c}{4\pi}\E\times (\dfrac{c}{\mu\omega}{\bf k}\times \E)=\dfrac{c}{4\pi}\sqrt{\dfrac{\epsilon}{\mu}}\E\times {\bf s}\times \E,
\]
where ${\bf k}=\dfrac{\omega}{v}{\bf s}$, ${\bf s}$ a unit vector, and $v=\dfrac{c}{\sqrt{\epsilon \mu}}$.
Using the form of the incident wave \eqref{eq:incidentfieldkappa<0}, the amount of energy $J^i$ of the incident wave $\E^i$ flowing through a unit area of the boundary per second,
 is then 
\[
J^i=|\S^i|\cos \theta_i=\dfrac{c}{4\pi}\sqrt{\dfrac{\epsilon_1}{\mu_1}} |\E_0^i|^2\cos \theta_i.
\]
Similarly, the amount of energy in the reflected and transmitted waves (also given in the previous section) leaving a unit area of the boundary per second is given by
\begin{align*}
J^r&=|\S^r|\cos \theta_i=\dfrac{c}{4\pi}\sqrt{\dfrac{\epsilon_1}{\mu_1}}|\E_0^r|^2\cos \theta_i\\
J^t&=|\S^t|\cos \theta_t=\dfrac{c}{4\pi}\sqrt{\dfrac{\epsilon_2}{\mu_2}}|\E_0^t|^2\cos \theta_t.
\end{align*}
The reflection and transmission coefficients are defined by
\[
\mathcal R=\dfrac{J^r}{J^i}=\left( \dfrac{|\E_0^r|}{|\E_0^i|}\right)^2,
\text{ and }\mathcal T=\dfrac{J^t}{J^i}=\sqrt{\dfrac{\epsilon_2\mu_1}{\epsilon_1\mu_2}}\dfrac{\cos \theta_t}{ \cos \theta_i}\left( \dfrac{|\E_0^t|}{|\E_0^i|}\right)^2.
\]
By conservation of energy or by direct verification, we have $\mathcal R + \mathcal T=1$.

For the case when no polarization is assumed, 
we have from Fresnel's equations \eqref{eq:fresnelformulaourwriting} that
\[
|\E_0^r|^2=R_{\parallel}^2+ R_{\perp}^2
=\left[\dfrac{ (z_2\,  x -z_1\,m)\cdot (x-\kappa \,m)}{ (z_2 \,  x + z_1\,m)\cdot (x-\kappa \,m)}\right]^2\,I_\parallel^2 +  
\left[\dfrac{ (z_1\,x -z_2 \, m)\cdot (x-\kappa \,m)}{(z_1\,x + z_2 \, m)\cdot (x-\kappa \,m)}\right]^2\,I_\perp^2,
\]
and so
\begin{align*}
\mathcal R
&=
\left(\dfrac{|\E_0^r|}{|\E_0^i|}\right)^2=
\dfrac{R_{\parallel}^2+ R_{\perp}^2}{I_{\parallel}^2+ I_{\perp}^2}\\
&=
\left[\dfrac{ (z_2\,  x -z_1\,m)\cdot (x-\kappa \,m)}{ (z_2 \,  x + z_1\,m)\cdot (x-\kappa \,m)}\right]^2\,\dfrac{I_\parallel^2}{I_{\parallel}^2+ I_{\perp}^2} +  
\left[\dfrac{ (z_1\,x -z_2 \, m)\cdot (x-\kappa \,m)}{(z_1\,x + z_2 \, m)\cdot (x-\kappa \,m)}\right]^2\,\dfrac{I_\perp^2}{I_{\parallel}^2+ I_{\perp}^2}\\
&=
 \left[ \dfrac{\left( \sqrt{\dfrac{\mu_2}{\epsilon_2}}-\sqrt{\dfrac{\epsilon_2\mu_2}{\epsilon_1^2}}\right)-\left( \sqrt{\dfrac{\mu_1}{\epsilon_1}}-\sqrt{\dfrac{\mu_2^2}{\epsilon_1\mu_1}}\right)x\cdot m}{\left( \sqrt{\dfrac{\mu_2}{\epsilon_2}}+\sqrt{\dfrac{\epsilon_2\mu_2}{\epsilon_1^2}}\right)+\left( \sqrt{\dfrac{\mu_1}{\epsilon_1}}+\sqrt{\dfrac{\mu_2^2}{\epsilon_1\mu_1}}\right)x\cdot m} \right]^2\,\dfrac{I_\parallel^2}{I_{\parallel}^2+ I_{\perp}^2}\\
 &\qquad +  
\left[ \dfrac{\left( \sqrt{\dfrac{\mu_1}{\epsilon_1}}-\sqrt{\dfrac{\mu_2^2}{\epsilon_1\mu_1}}\right)-\left( \sqrt{\dfrac{\mu_2}{\epsilon_2}}-\sqrt{\dfrac{\epsilon_2\mu_2}{\epsilon_1^2}}\right)x\cdot m}{\left( \sqrt{\dfrac{\mu_1}{\epsilon_1}}+\sqrt{\dfrac{\mu_2^2}{\epsilon_1\mu_1}}\right)-\left( \sqrt{\dfrac{\mu_2}{\epsilon_2}}-\sqrt{\dfrac{\epsilon_2\mu_2}{\epsilon_1^2}}\right)x\cdot m} \right]^2\,\dfrac{I_\perp^2}{I_{\parallel}^2+ I_{\perp}^2}
\end{align*}
which is a function only of $x\cdot m$.
We assume in these formulas that the product $\epsilon_i\mu_i>0$, $i=1,2$.
In principle the coefficients $I_\parallel$ and $I_\perp$ might depend on the direction $x$, in other words, for each direction $x$ we would have a wave that changes its amplitude with the direction of propagation.
The energy of the incident wave would be $f(x)=|\E_0^i|^2=I_\parallel(x)^2+I_\perp(x)^2$.
Notice that if the incidence is normal, that is, $x=m$, then
$\mathcal R=\left( \dfrac{z_2-z_1}{z_1+z_2}\right)^2$ which shows that even for radiation normal to the interface we may lose energy by reflection.

\begin{remark}\rm
The Fresnel coefficients $\mathcal R$ and $\mathcal T$ being written in terms of $x$ and $m$ 
as above are useful for studying some refraction problems in geometric optics. 
In fact, when $\kappa >0$, it was shown in \cite{CG2013} that it is possible to construct a surface interface between media $I$ and $II$ in such a way that both the incident radiation and transmitted energy are prescribed, and that loss of energy due to internal reflection across the interface is taken into account. 
\end{remark}

\section{Conclusion}
We have given a vector formulation of the Snell law for waves passing between two homogeneous and isotropic materials when one of them has negative refractive index, that is, is left-handed. This formulation was used to find surfaces having the uniform refraction property both in the far field and near field cases.
In the near field case, in contrast with the case when both materials are standard, these surfaces can be neither convex nor concave and can wrap around the target. A quantitative  analysis in terms of the parameters defining the surfaces has been carried out.
We have used the vector formulation of Snell's law to find expressions for the Fresnel formulas and coefficients for NIMs.
We expect these formulas to be useful in the design of surfaces separating materials, one of them a NIM, that refract radiation with prescribed amounts of energy given in advance. This will be done in future work and requires more sophisticated mathematical tools.

\end{document}